%% file: root-TCNS-final.tex
\pdfoutput=1
\documentclass[journal,10 pt,twoside,web]{ieeecolor} 
\input{preamble.tex}
\input{preamble_mkazma.tex}
\usepackage{lcsys} 
\usepackage{xcolor} 

\hypersetup{
pdfauthor={Mohamad H. Kazma},
pdftitle={Observability for Nonlinear Systems}
}

\titlesc{Observability for Nonlinear Systems:\\ Connecting Variational
Dynamics, Lyapunov Exponents, and Empirical Gramians}

\author{Mohamad H. Kazma, \textit{Graduate Student Member, IEEE} and Ahmad F. Tah$\text{a}^{\diamond}$, \textit{Member, IEEE}\vspace{-0.3cm}
	\thanks{
		$^\diamond$Corresponding author. This work is supported by National Science Foundation under Grants 2152450 and 2151571. The authors are with the Civil $\&$ Environmental Engineering and Electrical $\&$ Computer Engineering Departments, Vanderbilt University, 2201 West End Ave, Nashville, Tennessee 37235. Emails: mohamad.h.kazma@vanderbilt.edu, ahmad.taha@vanderbilt.edu.}
}

\begin{document}
\newdimen\origiwspc\newdimen\origiwstr\origiwspc=\fontdimen2=\fontdimen3
\makeatletter
\def\ps@arxivheadings{
\def\@oddfoot{}\def\@evenfoot{}
\def\@oddhead{\hbox{}\scriptsize\rightmark\hfil\thepage}
\def\@evenhead{\scriptsize\thepage\hfil\leftmark\hbox{}}}
\makeatother
\thispagestyle{arxivheadings}
\pagestyle{arxivheadings}
\markboth{IEEE Transactions on Control of Network Systems, in press, August 2026}{IEEE Transactions on Control of Network Systems, in press, August 2026}
\begin{abstract} 
	Observability quantification is a key problem in dynamic network sciences. While it has been thoroughly studied for linear systems, observability quantification for nonlinear networks is less intuitive and more cumbersome. One common approach to quantify observability for nonlinear systems is via the \textit{Empirical Gramian} (Empr-Gram)---a generalized form of the Gramian of linear systems. In this paper, we produce three new results. First, we establish that a variational form of discrete-time autonomous nonlinear systems yields a so-called \textit{Variational Gramian} (Var-Gram) that is equivalent to the classic Empr-Gram under linear output mappings; the former being easier to compute than the latter. Via \textit{Lyapunov exponents} derived from Lyapunov's direct method, the paper's second result derives connections between existing observability measures and Var-Gram. The third result demonstrates the applicability of these new notions for sensor selection placement in nonlinear systems. Numerical case studies demonstrate these three developments and their merits.
\end{abstract}
\vspace{-0.2cm}

\vspace{-0.6cm}
\section{Introduction and Contributions}\label{sec:Intro}
\IEEEPARstart{O}{bservability} is most generally defined as the ability to reconstruct the state variables of a dynamic system from limited output measurements~\cite{Kawano2013}. 
For deterministic~\cite{Kalman1963} and stochastic~\cite{Aoki1968} linear systems, quantifying observability is well-established.
However, the direct extension of observability notions from linear to nonlinear systems is not straightforward. We briefly discuss this literature next. 

An analytic differential approach introduced in~\cite{Hermann1977} evaluates observability by computing the \textit{Lie derivatives} around an initial point. 
Lie derivatives are typically avoided in practice for two reasons. $(i)$ Lie derivatives are computationally expensive and require the calculation of higher order derivatives~\cite{Whalen2015a}, and $(ii)$ the resulting observability measure is a rank condition that is qualitative in nature~\cite{Krener2009} and difficult to optimize. That is, the quantification is binary and thus does not lend itself easily to optimization problems such as sensor selection.

Other formulations can also be utilized to assess a nonlinear system's observability. 
One method follows from formulating the \textit{empirical observability Gramian} (Empr-Gram) of the system by considering an impulse response approach~\cite{Moore1981,Lall1999,Lall2002,Kunwoo2023}.
However, properly scaling the internal state variables such that the Gramian's eigenvalues accurately capture local variations in the states is not straightforward~\cite{Krener2009}. Furthermore, the Empr-Gram can be extended to account for stochastic nonlinear systems from a statistical point of view, as proposed in~\cite{Powel2020b, Kunwoo2023}.
A moving horizon approach for discretized nonlinear dynamics is introduced in~\cite{Haber2018} and further developed in~\cite{Kazma2023f}; it is based on a moving horizon formulation and offers a more robust solution than the Empr-Gram. The proposed approach does not establish a direct relation to the linear observability Gramian and the notions of Lyapunov stability. 

Lyapunov's second method is considered as a basis for observability of linear systems, yet a general theory to formulate Lyapunov functions for nonlinear systems is lacking~\cite{Liu2016a}. In the fields of chaos and ergodicity, a well-known method for assessing the stability of nonlinear system trajectory stems from Lyapunov's direct method on stability~\cite{AleksandrMikhailovichLyapunov1892}.
The direct method provides a characteristic spectrum of \textit{Lyapunov exponents} (LEs) that yields a basis for exponential asymptotic stability of dynamical systems.
This often-overlooked notion of stability in the field of control theory has recently been investigated in several areas. This includes studies that are related to bounds and observer design for linear time-varying (LTV) systems~\cite{Czornik2019} and model predictive control~\cite{Krishna2023}. 

An important aspect of the aforementioned stability method is that it is based on a \textit{variational} representation of the dynamical system~\cite{Barreira1998}, which means that the system is described by the evolution of infinitesimal state variations along its trajectory. Such variational system representations are considered for a wide class of nonlinear systems~\cite{Cortes2005}. The variational system can be viewed as an LTV model constructed along the tangent space of the nonlinear system~\cite{Kawano2021}, rendering the computation of an observability Gramian more efficient. In a recent study, considering the variational system of the general state-space formulation introduced in~\cite{Cortes2005}, an \textit{empirical differential Gramian} is formulated for the continuous-time domain~\cite{Kawano2021}. The introduced Gramian is similar to the Empr-Gram~\cite{Moore1981,Lall2002}; both are based on impulse response and empirical data, rendering both formulations computationally intensive. The impulse response around an initial state results in a fixed state trajectory. As such, the aforementioned observability results are considered local state trajectory dependent.

In this paper, we introduce a nonlinear observability Gramian that is based on a discrete-time variational representation. Our first objective is to illustrate that the proposed \textit{variational observability Gramian} (Var-Gram) is equivalent to Empr-Gram in characterizing local observability; the former being more computationally efficient. We also show that the Var-Gram reduces to the Gramian for a linear time-invariant (LTI) system. The second objective is to illustrate the connections relating LEs and Var-Gram measures. Based on this, we introduce conditions for observability of general nonlinear systems. The third objective is to showcase how these aforementioned developments can be applied to efficiently solve the sensor node selection (SNS) problem for nonlinear systems, which heavily relies on quantifying observability.

\parstartc{Paper contributions} The main contributions of this paper are three-fold. $(i)$ We formulate a new method for computing the observability Gramian of nonlinear systems with no inputs. This method is derived from a variational system representation of discrete-time nonlinear dynamics. We show and provide evidence that the proposed Var-Gram is equivalent to the Empr-Gram under linear output mappings, with extension to nonlinear outputs discussed in Appendix~\ref{apndx:generaloutput}; it reduces to the linear Gramian for a stable linear time-invariant (LTI) system. $(ii)$ We show that observability measures under Var-Gram are equivalent to LEs. 
We derive a local observability condition for nonlinear systems that is based on the spectral radius of the proposed Gramian. Such condition offers a bound for the observability of the systems in relation to LEs. $(iii)$ We show that the Var-Gram is a modular set function with respect to the binary sensor allocation vector, and that specific observability measures based on the proposed Var-Gram are submodular.
This submodularity enables the solution of the SNS problem in nonlinear networks to be scalable. This is analogous to SNS in linear networks where the submodularity of the linear Gramian is well-established.

\parstartc{Broader impacts} Establishing a connection between nonlinear observability and the LEs allows us to leverage the plethora of data-driven methods for computing LEs; see~\cite{Balcerzak2020} and references therein. Based on such methods, nonlinear observability can be evaluated from a data-driven perspective. Furthermore, the link between LEs and dynamical properties such as entropy~\cite{Young2013}, allow for studying observability in stochastic and chaotic systems. Investigating these prospects is outside the scope of this paper and merits future work.

\parstartc{Paper organization} As illustrated in Fig.~\ref{fig:framework}, the paper is organized as follows:~Section~\ref{sec:prelims} introduces the problem formulation and provides preliminaries on nonlinear observability. Section~\ref{sec:ObsGram} develops the theory behind the construction of the proposed Var-Gram. The connection between the Var-Gram and LEs is presented in Section~\ref{sec:LyapExp}. Section~\ref{sec:SNS1} presents the Var-Gram properties for the SNS problem. The numerical results are presented in Section~\ref{sec:casestudies}. Section~\ref{sec:conclusion} concludes this paper. 

\parstartc{Notation} Let $\mathbb{N}$, $\mathbb{R}$, $\mathbb{R}^n$, and $\mathbb{R}^{p\times q}$ denote the set of natural numbers, real numbers, real-valued column vectors of size $n$, and $p$-by-$q$ real matrices respectively. The symbol $\otimes$ denotes the Kronecker product. The cardinality of a set $\mc{N}$ is denoted by $|\mc{N}|$. The set $\emptyset$ represents the empty set. The operators $\mr{det}(\mA)$, $\mr{rank}(\mA)$, $\mr{trace}(\mA)$, and $\lambda(\mA)$ return the determinant, rank, trace, and the vector of eigenvalues of matrix $\mA$. Let $\m{I}_n \in \mathbb{R}^{n \times n}$ represent the identity matrix of size $n$. For a positive semidefinite (PSD) matrix $\m{A} \succeq 0$, $\log(\m{A})$ denotes the matrix logarithm. The operator $\{\m x_{i}\}_{i=0}^{\mr{N}} \in \mathbb{R}^{\mr{N}n}$ constructs a column vector that concatenates vectors $\m{x}_i \in \mathbb{R}^{n}$ for all $i \in \{0,1, \cdots, \mr{N}\}$. The operator $\mr{diag}\{\alpha_i\}_{i=1}^{n} \in \mathbb{R}^{n \times n}$ constructs a diagonal matrix from the scalars $\alpha_i$, while $\mr{diag}\{\m x\} \in \mbb{R}^{n \times n}$ constructs a diagonal matrix from vector $\m x\in  \mbb{R}^{n}$. The dot-product of two matrix-valued vectors $\m{\xi}$ is represented as $\langle\m{\xi},\m{\xi}\rangle:= \m{\xi}^{\top}\m{\xi}$, where the superscript $\top$ denotes the transpose. The $\mc{L}_2$-$\mr{norm}$ of vector $\m{x}$ is denoted by $||\m{x}||_{2} := \sqrt{\langle\m{x},\m{x}\rangle}$. For a matrix $\m{A}$, $||\m{A}||$ denotes the induced $\mc{L}_2$-$\mr{norm}$. The operator $\m{h} \circ \m{f} := \m{h}\left(\m{f}(\m{x})\right)$ denotes the composition of functions. The notation $ t_0$ and $t$ (as subscripts and superscripts) of a flow map are used for continuous-time mappings, while $0$ and $k$ are used for discrete-time mappings.

\begin{figure*}[t]
	\centering
	\scriptsize
	\resizebox{1\textwidth}{!}{
		\begin{tikzpicture}[
			node distance=0.6cm and 0.95cm,
			box/.style={draw, rounded corners, fill=blue!5, minimum height=1.3cm, minimum width=2.5cm, align=center},
			arrow/.style={->,>=stealth, thick,shorten <=5pt, shorten >=5pt},
			label/.style={font=\normal, align=center}
			]
			
\node[box, fill=Sub1purple, minimum width=4cm, minimum height=1.6cm] (setup) {};
			\node[font=\scriptsize, align=center] at ($(setup.north)+(0,-0.3)$) {\textbf{System dynamics}\\\textit{Section~\ref{sec:prelims}}};
			
\node[fill=Sub1purple!60!Sub1color, draw=black, rounded corners=2pt, font=\scriptsize, minimum width=1.5cm, minimum height=0.7cm] (setupLeft) at ($(setup.center)+(-1,-0.35)$) {Continuous~\eqref{eq:model_CT}};
			\node[fill=Sub1purple!60!Sub1color, draw=black, rounded corners=2pt, font=\scriptsize, minimum width=1.5cm, minimum height=0.7cm] (setupRight) at ($(setup.center)+(1,-0.35)$) {Discrete~\eqref{eq:model_DT}};
			\draw[->, thick, >=stealth] (setupLeft.east) -- (setupRight.west);
			
\node[box, fill=Sub1purple, minimum width=3.5cm] (obs) [below=0.2cm of setup] {
				\textbf{Nonlinear observability}\\[1pt]
				Definitions~\ref{def:Hermkern},~\ref{def:Hermkern2}, and~\ref{def:uniformobs}\\[1pt]
				\textit{Section~\ref{sec:Obs}}
			};
			\draw[thick, double distance=2pt, double=Sub1purple!90!Sub1color, -, shorten >=0pt, shorten <=0pt] ($(setup.south)+(0,0)$) -- ($(obs.north)+(0,0)$);
			
\begin{scope}[on background layer]
				\node[draw=gray, fill={rgb,255:red,247; green,247; blue,247}, rounded corners, inner sep=4pt, fit=(setup)(setupLeft)(setupRight)(obs)] (P1) {};
			\end{scope}
			
\node[box, fill=Sub1color] (p2) at ($(setup)+(5.4cm,0.15cm)$) {
				\textbf{Empr-Gram}\\[1pt]
				$\displaystyle \mathbf{W}_o^{\varepsilon}=\dfrac{1}{4\varepsilon^2}\sum_{k}(\Delta\mathbf{Y}_k^{\varepsilon})^{\top}\Delta\mathbf{Y}_k^{\varepsilon}$\\[1pt]
				\textit{Section~\ref{sec:Empr}}
			};
			
			\node[box, fill=Sub1color] (submod) [below= 0.5cm of p2] {
				\textbf{Var-Gram}\\[1pt]
				$\m{V}_{o} = {\m{\Psi}(\m{x}_{0})}^{\top}\m{\Psi}(\m{x}_{0})$\\[1pt]
				\textit{Section~\ref{sec:ObsGram}}
			};
			
\draw[thick, double distance=2pt, double=Sub1color!90!Sub1color, -, shorten >=0pt, shorten <=0pt] 
			($(p2.south)+(0.2,0)$) -- ($(submod.north)+(0.2,0)$)
			node[midway, left=12pt, font=\scriptsize] {Equivalence};
			\draw[thick, double distance=2pt, double=Sub1color!0!Sub1color, -, shorten >=0pt, shorten <=0pt] 
			($(p2.south)+(-0.2,0)$) -- ($(submod.north)+(-0.2,0)$)
			node[midway, right=0.2pt, font=\scriptsize] {in \hspace*{0.12cm}Theorem~\ref{theo:Equivelence}};
			
			\begin{scope}[on background layer]
				\node[draw=gray, fill={rgb,255:red,247; green,247; blue,247}, rounded corners, inner sep=4pt, fit=(p2)(submod)] (P2) {};
			\end{scope}
			
\node[box, fill=Sub1color] (p3) [right= 1.7cm of submod] {
			\vspace{0.2cm} \\	\textbf{SNS problem $\mb{P1}$}\\[1pt]
				$\underset{\mc{S}\subseteq\mc{V},\; |\mc{S}|=r}{\max} \;\; \mc{O}(\mc{S})$\\[2pt]
				\textit{Section~\ref{sec:SNS1}}
			};
			\node[box, fill=Sub1color, anchor=north, font=\scriptsize, align=center,minimum height=0.5cm] 
			(p3Label) at ($(p3.north)+(0,0.4)$) 
			{\textbf{Submodular measures} $\mc{O}(\mc{S})$
			};

\node[box, fill=Sub1color] (submod2) [above= 0.5cm of p3] {
				\textbf{LEs and observability}\\[1pt]
				Theorems~\ref{theo:logdet} and~\ref{theo:spectral-radius}\\[1pt]
				\textit{Section~\ref{sec:LyapExp}}
			};
			
			\begin{scope}[on background layer]
				\node[draw=gray, fill={rgb,255:red,247; green,247; blue,247}, rounded corners, inner sep=4pt, fit=(p3)(submod2)] (P3) {};
			\end{scope}
			
\draw[thick, double distance=1.8pt, double=Sub1color!90!Sub1color, -, shorten >=0pt, shorten <=0pt]
			($(submod.east)+(0,0.4)$) to[out=0, in=180] ($(submod2.west)+(0,-0.2)$);
			
			\draw[thick, double distance=1.8pt, double=Sub1color!0!Sub1color, 
			-, shorten >=0pt, shorten <=0pt] 
			($(submod.east)+(0,-0.2)$) 
			to[out=0, in=265] 
			($(p3Label.south)+(-1.5,0)$);
			
			\draw[thick, double distance=2pt, double=Sub1purple!0!Sub1color, -, shorten >=0pt, shorten <=0pt] ($(submod2.south)+(0,0)$) -- ($(p3Label.north)+(0,0)$);

\node[box, fill=Sub1purple, right=of P3, minimum width=4cm] (sim) {
				\textbf{Numerical case studies}\\[1pt]
				Illustration of theoretical results\\
				Solving SNS problem $\mb{P1}$\\[1pt]
				\textit{Section~\ref{sec:casestudies}}
			};
			
\node[draw=gray!50, dashed, thick, fit=(P2)(P3), inner sep=5pt, name=PartitionToSP] {};
			
\draw[arrow] (P1) -- (PartitionToSP);
\draw[arrow] (PartitionToSP) -- (sim);
		\end{tikzpicture}
	}
	\vspace{-0.5cm}
	\caption{Illustration of nonlinear observability quantification and the corresponding optimal sensor selection. Empirical and variational Gramians (Sections~\ref{sec:Empr},~\ref{sec:ObsGram}) are shown to be equivalent (Theorem~\ref{theo:Equivelence}). The Var-Gram is connected to LEs and observability (Section~\ref{sec:LyapExp}). The Var-Gram results in a submodular sensor selection problem $\mb{P1}$ (Section~\ref{sec:SNS1}), illustrated through case studies (Section~\ref{sec:casestudies}).}
	\label{fig:framework}
	\vspace{-0.6cm}
\end{figure*}
\vspace{-0.1cm}
\section{Preliminaries and Definitions}\vspace{-0.1cm}\label{sec:prelims}
Let the following represent a general continuous-time nonlinear dynamic network without input.
\begin{equation}\label{eq:model_CT}
	\dot{\m{x}}(t) = \m f(\m{x}(t)), \; \quad
	\m{y}(t) =  \m{h}({\m{x}}(t)),
\end{equation}
where the smooth manifold $\mc{M}\subseteq \mathbb{R}^{n_x}$ represents the state-space under the action of system dynamics, the system state vector evolving in $\mc{M}$ is denoted as $\m x(t):= \m{x} \in\mathbb{R}^{n_x}$, and $\m y(t) := \m{y}\in \mathbb{R}^{n_{y}}$ is the global output measurement vector.
The nonlinear mapping function $\m{f}(\cdot):\mc{M} \rightarrow\mbb{R}^{n_x}$ and nonlinear mapping measurement function $\m{h}(\cdot) :\mc{M} \rightarrow\mbb{R}^{n_y}$ are smooth and at least twice continuously differentiable. 
\begin{asmp}\label{assumption:compact}
There exists a compact set $\mc{\m{X}} \subseteq \mc{M}$ such that, for any system initialization $\m{x}_{0} \in \mc{\m{X}}_{0}$, the system trajectory remains in $\mc{\m{X}}$ for all $t \ge t_0 = 0$. The compact set $\mc{\m{X}}$ contains all the feasible system trajectories.
\end{asmp}

This assumption dictates that $\m{x}$ belongs to a compact set $\mc{X}$ along the system trajectory. This is not restrictive when considering nonlinear networks such as cyber-physical networks with bounded state $\m{x}$. We introduce the continuous model~\eqref{eq:model_CT} to establish nonlinear observability concepts in Section~\ref{sec:Obs}. This enables the derivation of variational observability conditions and Proposition~\ref{prs:discrete_var}. Furthermore, numerical discretization of the continuous-time model renders a system amenable towards various SNS applications.

To that end, we consider a discrete-time nonlinear dynamical system without inputs, representing the action of the system dynamics evolving on the smooth manifold $\mathcal{M}$. We adopt the \textit{implicit Runge-Kutta} (IRK) method~\cite{Iserles2009} as the discretization method. While there are myriad discretization methods for general nonlinear systems, IRK is utilized since it offers a wide range of applicability to systems with varying degrees of stiffness; refer to~\cite{Atkinson2011}. 
The continuous-time system~\eqref{eq:model_CT} can be expressed in the following discrete-time form.
\begin{subequations}\label{eq:model_DT}
	\begin{align}
		\m{x}_{k+1} &= \m x_{k} +\tilde{\m{f}}(\m{x}_{k}),\label{eq:RKutta_dynamics_compact}\\
		\m{y}_{k} &=  \m{h}(\m{x}_{k}),\label{eq:DT_measurement_model}
\end{align}
\end{subequations}
where $k\in\mbb{N}$ is the discrete-time index such that $\m x_k = \m x(kT)$ and $T > 0$ denotes the discretization period.
The nonlinear mapping function $\tilde{\m{f}}(\cdot)\in \mathbb{R}^{n_x}$ represents the dynamics depicted by $\m{f}(\cdot)$ under the action of a discrete-time model. The nonlinear mapping function $\tilde{\m{f}}(\cdot)$ is defined for the IRK method as $\tilde{\m{f}}(\m{x}_k):= \tfrac{T}{4}\left(\m f(\m \zeta_{1,k+1})+3\m f(\m \zeta_{2,k+1})\right)$. Vectors $\m \zeta_{1,k+1},\m \zeta_{2,k+1}\in\mbb{R}^{n_x}$ are auxiliary for computing $\m x_{k+1}$ provided that $\m x_{k}$ is given. Refer to Appendix~\ref{apdx:Runge-Kutta} for details on the IRK method's auxiliary vectors that are essential for analytically computing the Var-Gram.
\vspace{-0.2cm}
\subsection{Observability of nonlinear systems}\vspace{-0.1cm}\label{sec:Obs}
To make this paper self-contained, we briefly recall the notions of nonlinear observability. 
\subsubsection{Flow maps and output mapping} Consider $\m{\phi}_{t_0}^{t} : \mc{M} \rightarrow \mc{M}$ which maps a state $\m{x}_{0} \in \mc{{X}}_0$ at initial time $t_0$ to a state $\m{x} \in \mc{{X}}$ at time $t>t_0$~\cite{Shadden2005a, Iserles2009}. The continuous-time flow map of nonlinear system~\eqref{eq:model_CT} can be written as
\begin{equation}\label{eq:flow_CT}
	\m{\phi}_{t_0}^{t}(\m{x}_{0}) := \m{x}(t)=\m{x}(t_0) +
	\int_{t_0}^{t} \m{f}(\m{x}(\tau))d\tau.
\end{equation}

This interpretation of the nonlinear dynamical system~\eqref{eq:model_CT} defines the flow of phase points along the phase curve, thereby describing the time evolution of the state variables pertaining to the dynamical system~\cite{Arrowsm1994}.
\subsubsection{Distinguishability and local observability} Observability of the nonlinear dynamical system is then defined as the ability to identify the initial states $\m{x}_{0}$ satisfying Assumption~\ref{assumption:compact} for $t > t_0$. Definitions~\ref{def:Hermkern} and \ref{def:Hermkern2} are based on the notions of observability from the work of Hermann and Krener~\cite{Hermann1977}.
\begin{mydef}[distinguishability~\cite{Hermann1977}]\label{def:Hermkern}
	Any two initial states $\m{x}_0$ and $\m{x}_1 \in \mc{\m{X}}_0$ are \textit{indistinguishable} if and only if for any $\m{\phi}^{t}_{t_0} \in \mc{\m{X}}$, we have $\m{h}\circ\m{\phi}^{t}_{t_0}(\m{x}_0) = \m{h}\circ\m{\phi}^{t}_{t_0}(\m{x}_1) \; \forall \; t > t_0$.
\end{mydef}
\begin{mydef}[local observability~\cite{Hermann1977}]\label{def:Hermkern2}
	A nonlinear system is \textit{locally weakly observable} at $\m{x}_0$ if there exists a neighborhood $\mc{D} \in \mc{\m{X}}_0$ such that $\m{h}\circ\m{\phi}_{t_0}^{t}(\m{x}_{0}) \neq \m{h}\circ\m{\phi}_{t_0}^{t}({\m{x}_1})$ for $t > t_0$ for all $\m{x}_0 \neq \m{x}_1 \in \mc{D}$, i.e., if the compositions are distinguishable. The system is said to be \textit{locally observable} if it is \textit{locally weakly observable} for all $\m{x}_0 \in \mc{X}_0$.
\end{mydef}

The above definitions\footnote{Frankly, these are theoretical results and more than just \textit{definitions} but we stick to their usage as definitions to comply with the recent literature.} are essential for the subsequent observability definitions and the derivation of the observability conditions in Section~\ref{sec:LyapExp}.

\subsubsection{Notion of uniform observability}\label{subsec:uniformobs}
A seemingly stronger definition of observability (Definition~\ref{def:Hermkern2}) for continuous systems is investigated in~\cite{Hanba2017}. It shows that an observation window of finite length $\mr{N} \in \mbb{N}$ exists given the realization of distinguishability (Definition~\ref{def:Hermkern}). Such sequence of measurements uniquely determines the initial state $\m{x}_0$ on the compact set $\mc{{X}}_{0}$. The existence of such finite-observation window defines the \textit{uniform observability} criterion for discrete-time nonlinear systems~\cite{Hanba2009}. 
\begin{mydef}[uniform observability~\cite{Hanba2009}]\label{def:uniformobs}
	The system~\eqref{eq:model_DT} is said to be \textit{uniformly observable} over compact set $\mc{{X}}$ if there exists a finite-observation window $\mathrm{N} \in \mbb{N}$ such that the output sequence 
	\begin{equation}\label{eq:UniformObs}
		\m{\xi}(\m{x}_0) := \big \{\m{y}_{k} \big \}_{k=0}^{\mr{N}-1}
		\in \mathbb{R}^{\mr{N}n_y},
		\vspace{0.1cm}
	\end{equation}
	is injective  (one-to-one) with respect to $\m{x}_0\in \mc{\m{X}}_0$. This implies that the Jacobian of $\m{\xi}(\m{x}_0)$ is full rank, i.e., $ \mathrm{rank}\left(\tfrac{\partial\m{\xi}(\m{x}_0)}{\partial \m{x}_0}\right) = n_x\; \forall \; \m{x}_{0} \in \mc{X}_0.$ This rank condition is a sufficient condition for uniform observability in~$\mc{X}_0$, a consequence of the real Jacobian conjecture~\cite{SmaleSteve1998}.
\end{mydef}

The following theorem shows that distinguishability along with the observability rank condition are equivalent to uniform observability. 
\begin{theorem}[\hspace{-0.012cm}\cite{Hanba2009}, Th. 7]\label{theo:bijection}
		Consider the discrete-time nonlinear system~(2) and assume it satisfies Definition~\ref{def:Hermkern} (distinguishability). The system is then locally observable in the set $\mc{X}_0$ if there exists a finite horizon $\mr{N}_x < \mr{N}$ such that the Jacobian $\tfrac{\partial\m{\xi}(\m{x}_0)}{\partial \m{x}_0} = \left\{\tfrac{\partial \m{y}_{k}}{\partial \m{x}_{0}}\right\}_{k=0}^{\mr{N_x}-1}$ 
		is full rank, i.e., 
		$\mathrm{rank}\left(\tfrac{\partial\m{\xi}(\m{x}_0)}{\partial \m{x}_0} \right)= n_x$ 
		for all $\m{x}_0 \in \mc{X}_0$.
\end{theorem}

For system~\eqref{eq:model_DT}, the discrete-time equivalent of the flow map~\eqref{eq:flow_CT} can be denoted as $\m{\phi}_{0}^{k}(\m{x}_0) \equiv \m{\phi}_{0}^{kT}(\m{x}_0) = \m{x}_k$. The flow map in discrete-time is obtained by replacing the integral in~\eqref{eq:flow_CT} with a summation over the discrete mapping function updates. The composition mapping $\m{h}\circ\m{\phi}_{0}^{k}(\m{x}_0)$ represents the output measurement vector $\m{y}_k$ at time step $k$. The output sequence $\tfrac{\partial\m{\xi}(\m{x}_0)}{\partial \m{x}_0} \in \mathbb{R}^{\mr{N}n_y\times n_x}$ for observation horizon $\mr{N}$ is given by
\begin{equation}\label{eq:xi}
	\tfrac{\partial\m{\xi}(\m{x}_0)}{\partial \m{x}_0} \hspace{-0.05cm}=\hspace{-0.05cm}
	\Big\{\tfrac{\partial \m{y}_{k}}{\partial \m{x}_{0}}\Big \}^{\mr{N}-1}_{k=0}
	\hspace{-0.05cm}=\hspace{-0.05cm} \left\{\tfrac{\partial\m{h}\left(\m{\phi}_{0}^{k}(\m{x}_0)\right)}{\partial\m{\phi}_{0}^{k}(\m{x}_0)}\tfrac{\partial\m{\phi}_{0}^{k}(\m{x}_0)}{\partial \m{x}_0}\right\}^{\mr{N}-1}_{k=0},
\end{equation}
where $\tfrac{\partial\m{h}\left(\m{\phi}_{0}^{k}(\m{x}_0)\right)}{\partial\m{\phi}_{0}^{k}(\m{x}_0)} \in \mathbb{R}^{n_y \times n_x}$ is the Jacobian of the output function with respect to the state $\m{\phi}_{0}^{k}(\m{x}_0) = \m{x}_k$ and $\tfrac{\partial\m{\phi}_{0}^{k}(\m{x}_0)}{\partial{\m{x}_0}}\in \mathbb{R}^{n_x \times n_x}$ is the Jacobian of the flow map with respect to the initial state $\m{x}_0$. We note that Definition~\ref{def:uniformobs} and Theorem~\ref{theo:bijection} imply that a unique solution exists around $\m{x}_0$, thereby establishing local weak observability (Definition~\ref{def:Hermkern2})~\cite{Hanba2009,Haber2018}.

\vspace{-0.1cm}
\subsection{Empirical observability Gramian}\label{sec:Empr}
The Empr-Gram is one approach that can be used to quantify the observability of nonlinear systems; see Section~\ref{sec:Intro}.
The discrete-time Empr-Gram ${\m{W}}_{o}^{\eps}(\m{x}_{0})\in \mathbb{R}^{n_x \times n_x} $~\cite{Lall2002,Krener2009,Powel2015,Kunwoo2023} can be written as
\begin{equation}\label{eq:Emp0bsGram}
	\textit{Empr-Gram:}\;\;\;	{\m{W}}_{o}^{\eps}(\m{x}_{0}) := \frac{1}{4\eps^{2}}\sum_{k=0}^{\mr{N}-1}(\m{\Delta Y}^{\eps}_{k})^{\top}
	\m{\Delta Y}^{\eps}_{k},
\end{equation}
where the impulse response measurement vector $\m{\Delta Y}^{\eps}_{k}:=\m{\Delta Y}^{\eps}_{k}(\m{x}_{0}) = \bmat{\m{y}_k^{+1} - \m{y}_k^{-1},\; \cdots \; , \m{y}_k^{+n_x} - \m{y}_k^{-n_x}} \in \mathbb{R}^{n_y\times n_x}$ and $\m{y}_k^{\pm i} := \m{y}_k(\m{x}_{0}\pm \eps\m{e}_{i},i) = \m{h}(\m{\phi}_0^k(\m{x}_0 \pm \varepsilon \m{e}_i,i)) \in \mathbb{R}^{n_y}$ represents the system output at time step $k$ resulting from an infinitesimal perturbation of $\pm \eps$ applied to the $i$-th component of the initial state $\m{x}_0$. The superscript $i$ in $\m{y}_k^{\pm i}$ indicates that the perturbation is applied in the direction of the $i$-th standard basis vector, denoted by $\m{e}_i \in \mathbb{R}^{n_x}$, where $i \in \{1, 2, \cdots, n_x\}$. The Empr-Gram~\eqref{eq:Emp0bsGram} thus captures the system's response to infinitesimal positive and negative perturbations applied along each standard basis vector of the state space. The horizon for which the Empr-Gram~\eqref{eq:Emp0bsGram} is computed is $\mr{N}$. This choice is based on Definition~\ref{def:uniformobs} of uniform observability and Theorem~\ref{theo:bijection}. It is well-established that the system~\eqref{eq:model_DT} is locally weakly observable at an initial vector $\m{x}_{0} \in \mathcal{\m{X}}_0$ if $\lim_{\eps \rightarrow0}\mr{rank}(	{\m{W}}_{o}^{\eps}(\m{x}_{0})) = n_x$~\cite{Powel2015}.

Note that the Empr-Gram is equivalent to the observability rank condition~\cite{Kalman1963} for an LTI system. That is, for any time-invariant linear system $(\m{A},\; \m{C})$
\begin{equation}\label{eq:linear}
	\m{x}_{k+1} = \m{A}\m{x}_{k}, \; \quad \m{y}_{k}= \m{C}\m{x}_{k},
\end{equation}
the Empr-Gram reduces to the linear observability Gramian ${\m{W}}_{o}^{l}(\m{x}_{0})  := \langle{\m{O}}_{l},{\m{O}}_{l}\rangle \in \mathbb{R}^{n_x \times n_x}$, where the observability matrix is defined as ${\m{O}}_{l} := \bmat{
	\m{C}^{\top},\; 
	(\m{C}\m{A})^{\top},\; 
	\cdots,\; 
	(\m{C}\m{A}^{\mr{N}-1})^{\top}}^{\top} \in \mathbb{R}^{\mr{N}n_y \times n_x}$ for any $\m{x}_0 \in \mathcal{\m{X}}_0$ and any infinitesimal $\eps >0$. This result is established in~\cite[Lemma 7]{Lall1999} and~\cite[Lemma 5]{Lall2002}. In the subsequent sections, we refer to this established result, where we show that, similar to the Empr-Gram (which generalizes the linear Gramian for nonlinear systems), the Var-Gram exhibits analogous properties. For observability of LTI systems, see~\cite{Kalman1963,Liu2016a}. In the next section, we introduce the proposed Var-Gram for nonlinear systems.

\vspace{-0.1cm}
\section{A Variational Approach for Quantifying Nonlinear Observability}\vspace{-0.1cm}\label{sec:ObsGram}
We consider the variational (or prolonged, a term borrowed from~\cite{Cortes2005}) system along the continuous-time flow map $\m{\phi}_{t_0}^{t}$~\cite{Cortes2005,Kawano2021}. The variational system can be written in the following form
\begin{equation}\label{eq:VarCortes}
	\m{\delta}\m{\dot{x}}(t) =  \dfrac{\partial \m{f}(\m{\phi}_{t_0}^{t})}{\partial{\m{\phi}_{t_0}^{t}}}\m{\delta}\m{x}(t), \quad
	\m{\delta}\m{y}(t)=  \dfrac{\partial\m{h}\left(\m{\phi}_{t_0}^{t}\right)}{\partial\m{\phi}_{t_0}^t}
	\m{\delta}\m{x}(t), \vspace{+0.1cm}
\end{equation}
where $\m{\delta}\m{x}(t) \in \mathbb{R}^{n_x}$ is the asymptotic time evolution of the state variables at time $t$ along the system trajectory and $\m{\delta}\m{y}(t) \in \mathbb{R}^{n_y}$ is the corresponding variational measurement. Using the chain rule, the above dynamics can be written as the following infinitesimal variational system~\cite{Kawano2021,Martini2022}.
\begin{subequations}\label{eq:VarNonLinDynamics}
	\begin{align}
		\m{\delta}\m{x}(t)&= \m{\phi}_{t_0}^{t}(\m{x}_0+\m{\delta}\m{x}_0) - \m{\phi}_{t_0}^{t}(\m{x}_0) = \dfrac{\partial\m{\phi}_{t_0}^{t}(\m{x}_0)}{\partial{\m{x}_0}}\m{\delta}\m{x}_0,\\
		\m{\delta}\m{y}(t) &=  \dfrac{\partial\m{h}\left(\m{\phi}_{t_0}^{t}(\m{x}_0)\right)}{\partial\m{\phi}_{t_0}^t(\m{x}_0)}
		\dfrac{\partial\m{\phi}_{t_0}^{t}(\m{x}_0)}{\partial\m{x}_{0}}
		\m{\delta}\m{x}_0,
	\end{align}
\end{subequations}
where $\m{\delta}\m{x}_0 \in \mathbb{R}^{n_x}$ is the infinitesimal perturbation to the initial state $\m{x}_0$. Readers are referred to~\cite{Kawano2021} for the complete derivation of variational system~\eqref{eq:VarNonLinDynamics}. The following result provides a discrete-time analogue of variational system~\eqref{eq:VarNonLinDynamics}.

\begin{myprs}\label{prs:discrete_var}
	The infinitesimal discrete-time variational system representation of the continuous-time variational equations~\eqref{eq:VarNonLinDynamics} can be written as
	\begin{subequations}\label{eq:DiscVar}
		\begin{align}
\m{\delta}\m{x}_{k+1} &=
\m{\Phi}_{0}^{k}(\m{x}_0)\m{\delta}\m{x}_0 = \left(\m{I}_{n_x}  +\tfrac{\partial\tilde{\m{f}}(\m{x}_{k})}{\partial\m{x}_{k}}\right)\tfrac{\partial\m{x}_{k}}{\partial\m{x}_{0}}\m{\delta}\m{x}_0,\hspace{-0.2cm}\\
			\m{\delta}\m{y}_{k}
			&=
\m{\Psi}_{0}^{k}(\m{x}_{0})\m{\delta}\m{x}_{0} = \tfrac{\partial{\m{h}(\m{x}_{k})}}{\partial\m{x}_k}\m{\Phi}_{0}^{k}(\m{x}_0)\m{\delta}\m{x}_0,
\label{eq:var_measure}
		\end{align}
	\end{subequations}
	where $\m{\delta}\m{x}_{k+1} \in \mathbb{R}^{n_x}$ is the infinitesimal perturbation at time-index $k+1$ and $\m{\delta}\m{y}_{k} \in \mathbb{R}^{n_y}$ is the variational measurement at time-index $k$. The discrete-time variational mapping function is defined as $\m{\Phi}_{0}^{k}(\m{x}_0)
:= \left(\m{I}_{n_x}  +\tfrac{\partial\tilde{\m{f}}(\m{x}_{k})}{\partial\m{x}_{k}}\right)\tfrac{\partial\m{x}_{k}}{\partial\m{x}_{0}} \in \mathbb{R}^{n_x \times n_x}$ and matrix $\m{I}_{n_x}\in \mathbb{R}^{n_x \times n_x}$ is the identity matrix of size $n_x$. The variational measurement mapping function is denoted as $\m{\Psi}_{0}^{k}(\m{x}_{0}) :=  \tfrac{\partial{\m{h}(\m{x}_{k})}}{\partial\m{x}_k}\m{\Phi}_{0}^{k}(\m{x}_0)\in \mathbb{R}^{n_y\times n_x}$.
\end{myprs}
\begin{proof}
	Let $\hat{\m{x}}_0 = \m{x}_0 + \m{\delta}\m{x}_0 \in \mathcal{\m{X}}_0$, then for any discrete-time index $k>0$, the state-space equation~\eqref{eq:RKutta_dynamics_compact} can be written as $\hat{\m{x}}_{k+1} = \hat{\m{x}}_{k} + \tilde{\m{f}}(\hat{\m{x}}_{k} )$. Analogous to~\eqref{eq:VarNonLinDynamics}, the infinitesimal variational vector for index $k+1$ is rewritten as $\m{\delta}\m{x}_{k+1} = \m{\phi}_{0}^{k+1}(\hat{\m{x}}_0) - \m{\phi}_{0}^{k+1}(\m{x}_0) = \hat{\m{x}}_{k+1} - \m{x}_{k+1}.$ Applying the definition of Fréchet derivative~\cite[Definition 3.4.8]{Krantz2013}
	, i.e., directional derivative, and recalling that $\m{f}(\cdot)$ is at least twice differentiable, the following holds true: $\lim_{\m{\delta}\m{x}_0\rightarrow0}
	\frac{\hat{\m{x}}_{k+1} - \m{x}_{k+1}}{\hat{\m{x}}_0 -{\m{x}}_0 } =\frac{\hat{\m{x}}_{k+1} - \m{x}_{k+1}}{\m{x}_0 + \m{\delta}\m{x}_0-{\m{x}}_0}
	=\frac{\hat{\m{x}}_{k+1} - \m{x}_{k+1}}{\m{\delta}\m{x}_0} =\tfrac{\partial\m{x}_{k+1}}{\partial\m{x}_{0}}$. Similarly applying the Fréchet derivative to the nonlinear mapping function, we obtain: $\lim_{\m{\delta}\m{x}_0\rightarrow0}
	\frac{\tilde{\m{f}}( \hat{\m{x}}_{k}) - \tilde{\m{f}}(\m{x}_{k})}{\hat{\m{x}}_0 -{\m{x}}_0 } = \tfrac{\partial\tilde{\m{f}}(\m{x}_{k})}{\partial\m{x}_{0}}$; see Appendix~\ref{apdx:Runge-Kutta} 
for derivation of $\tfrac{\partial\tilde{\m{f}}(\m{x}_{k})}{\partial\m{x}_{0}}$. It follows that
	\begin{equation}
		\m{\delta}\m{x}_{k+1} = \underbrace{\hat{\m{x}}_{k} - \m{x}_{k}}_
		{\tfrac{\partial\m{x}_{k}}{\partial\m{x}_{0}}\m{\delta}\m{x}_0} +
		\underbrace{\tilde{\m{f}}( \hat{\m{x}}_{k}) - \tilde{\m{f}}(\m{x}_{k})}_
		{\tfrac{\partial\tilde{\m{f}}( \m{x}_{k})}{\partial\m{x}_{0}}\m{\delta}\m{x}_0},
	\end{equation}
	then by applying the chain rule, we can write
	$\tfrac{\partial\tilde{\m{f}}(\m{x}_{k})}{\partial\m{x}_{0}} =	\tfrac{\partial\tilde{\m{f}}(\m{x}_{k})}{\partial\m{x}_{k}}
	\tfrac{\partial\m{x}_{k}}{\partial\m{x}_{0}}$, the rest requires factoring out  $\tfrac{\partial\m{x}_{k}}{\partial\m{x}_{0}}$. Using an analogous approach, the variational measurement vector $\m{\delta}\m{y}_{k}$ for any $k>0$ can be written as 
\begin{equation}
		\m{\delta}\m{y}_{k} =
\hat{\m{y}}_{k} - \m{y}_{k} = \tfrac{\partial\m{y}_{k}}{\partial\m{x}_{0}}\m{\delta}\m{x}_k =
		\tfrac{\partial\m{h}(\m{x}_{k})}{\partial \m{x}_{k}}
		\tfrac{\partial \m{x}_{k}}{\partial\m{x}_0}\m{\delta}\m{x}_0, \end{equation}
	then, substituting the above with the variational state vector $\m{\delta}\m{x}_{k}$ we obtain~\eqref{eq:DiscVar}. This completes the proof.
\end{proof}

For ease of exposition, we write $ \m{x}_k \equiv \m{\phi}_{0}^{k} $ and treat the dependency on the initial condition $ \m{x}_0$ as implicit. That is, we drop the explicit argument $\m{x}_0$ in $ \m{\Phi}_{0}^{k}(\m{x}_0) $ and $ \m{\Psi}_{0}^{k}(\m{x}_0) $, denoting them as $ \m{\Phi}_{0}^{k} $ and $ \m{\Psi}_{0}^{k} $.

\begin{myrem}\label{rmk:Phi}
	The variational mapping function $\m{\Phi}_{0}^{k}$ requires the knowledge of states $\m{x}_i$  for $i \in \{0\;,1\;,\;\cdots\;,k\}$. Notice that this function can be written as $\m{\Phi}_{0}^{k} = \tfrac{\partial\m{\phi}_{0}^{k}}{\partial\m{\phi}_{0}^{k-1}}\tfrac{\partial\m{\phi}_{0}^{k-1}}{\partial\m{x}_{0}}$. Then, computing $\m{\Phi}_{0}^{k}$ requires evaluation under the chain rule and thus can be written for any discrete-time index $k$ as
	\begin{equation}\label{eq:Phi}
\m{\Phi}_{0}^{k}
		= \m{\Phi}_{k-1}^{k}  \m{\Phi}_{k-2}^{k-1}
		\; \cdots \; 	\m{\Phi}_{0}^{1}\m{\Phi}_{0}^{0}
= \prod^{i=k}_{1}\m{\Phi}^{i}_{i-1}
		,\hspace{-0.2cm}
	\end{equation}
	where matrix $\m{\Phi}_{i-1}^{i}= \tfrac{\partial\m{\phi}_{0}^{i}}{\partial\m{\phi}_{0}^{i-1}}$ represents the partial derivatives with respect to $\m{\phi}_{0}^{i-1}$. Note that $\m{\Phi}_{0}^{0} = \tfrac{\partial \m{x}_0}{\partial \m{x}_{0}} = \eye_{n_x}$; it is omitted for simplicity of exposition.
\end{myrem}
Based on variational system~\eqref{eq:DiscVar}, we now introduce the Var-Gram according to the following proposition.

\begin{myprs}\label{prs:proposed_observ_gramian}
	Consider a variational discrete-time system \eqref{eq:DiscVar} with measurement model. Following the results of Theorem~\ref{theo:bijection}, there exists a finite-time measurement horizon $\mr{N} \in \mathbb{N}$, such that the Var-Gram evaluated around initial state $\m{x}_{0} \in \mathcal{\m{X}}_0$ can be expressed as
\begin{equation}
		\boxed{\textit{Var-Gram:} \;\;\; \m{V}_{o}(\m{x}_{0}) := {\m{\Psi}(\m{x}_{0})}^{\top}\m{\Psi}(\m{x}_{0}) \in \mathbb{R}^{n_x \times n_x},}\label{eq:prop_obs_gram}
	\end{equation}
	where observability matrix $\m{\Psi}(\m{x}_{0}):= \m{\Psi} \in \mathbb{R}^{\mr{N}n_y \times n_x}$ concatenates the variational observations $\m{\delta}{\m{y}}_k$ over measurement horizon $\mr{N}$ for $k \in \{0,\;1,\;\cdots\;,\;\mr{N}-1\}$, and can be written as
	\begin{equation}\label{eq:prop_obs_matrix}
		\boxed{\m{\Psi} :=
				\bmat{ (\m{\Psi}_{0}^{0})^{\top}, (\m{\Psi}_{0}^{1})^{\top}, \ldots, (\m{\Psi}_{0}^{\mr{N}-1})^{\top} }^{\top}.} \vspace{+0.05cm}
	\end{equation}
\end{myprs}
Note that $\m{\Psi}_{0}^{k}$ is the variational measurement mapping function~\eqref{eq:var_measure}. For $k=0$, we have $\m{\Psi}_{0}^{0} = 	\tfrac{\partial{\m{h}(\m{x}_{0})}}{\partial\m{x}_{0}}\m{\Phi}_{0}^{0} = \tfrac{\partial{\m{h}(\m{x}_{0})}}{\partial\m{x}_{0}}$.
\begin{proof}
	Theorem~\ref{theo:bijection} implies the existence of an output sequence over a finite-time horizon, such that a unique solution around $\m{x}_0\in \mathcal{\m{X}}_0$ exists. With that in mind, we show that measurement vector~\eqref{eq:prop_obs_matrix} is equivalent to the Jacobian of the sequence $\m{\xi}(\m{x}_{0})$. This equivalence shows that Var-Gram~\eqref{eq:prop_obs_gram} represents the observability Gramian of the variational system. It follows from the definition of $\m{\Psi}_0^k$ in~\eqref{eq:var_measure} that, by applying the chain rule as discussed in Remark~\ref{rmk:Phi},~\eqref{eq:prop_obs_matrix} can be rewritten as
\begin{align} \label{prop2:eq1}
\m{\Psi} \hspace{-0.05cm}=\hspace{-0.05cm}
		\left\{\dfrac{\partial{\m{h}(\m{x}_{k})}}{\partial\m{x}_k}\m{\Phi}_{0}^{k}\right\}^{\mr{N}-1}_{k=0}
		\hspace{-0.05cm}=\hspace{-0.05cm}
		\left\{\dfrac{\partial{\m{h}(\m{x}_{k})}}{\partial\m{x}_k}
		\prod^{i=k}_{1}\m{\Phi}^{i}_{i-1}\right\}^{\mr{N}-1}_{k=0},
		\hspace{-0.2cm}
\end{align}
	we use column notation to express a matrix concatenated from iterated matrix entries for $k\in \{0\;,1\;,\;\cdots\;,\mr{N}-1\}$. Referring to Definition~\ref{def:uniformobs}, $\tfrac{\partial\m{\xi}(\m{x_0})}{\partial \m{x}_0}=\left\{\tfrac{\partial\m{h}\left(\m{\phi}_{0}^{k}(\m{x}_0)\right)}{\partial\m{\phi}_{0}^{k}(\m{x}_0)}\tfrac{\partial\m{\phi}_{0}^{k}(\m{x}_0)}{\partial \m{x}_0}\right\}^{\mr{N}-1}_{k=0}$; see~\eqref{eq:xi}.
The partial derivative $\tfrac{\partial\m{\phi}_{0}^{k}(\m{x}_0)}{\partial \m{x}_0}$ reduces to the Jacobian of the discrete-time state-equation~\eqref{eq:RKutta_dynamics_compact}. Under the action of the chain rule, it can be expressed as $\left(\m{I}_{n_x}  +\tfrac{\partial\tilde{\m{f}}(\m{x}_{k})}{\partial\m{x}_{k}}\right) \tfrac{\partial\m{\phi}_{0}^{k}(\m{x}_0)}{\partial \m{x}_0}$, which  is equal to $\m{\Phi}_{0}^{k}$. Thus, it follows that equations~\eqref{eq:prop_obs_matrix} and~\eqref{eq:UniformObs} are equivalent.
	This completes the proof.
\end{proof}
Based on Proposition~\ref{prs:proposed_observ_gramian}, we now show that the proposed Var-Gram~\eqref{eq:prop_obs_gram} for nonlinear systems reduces to the linear observability Gramian.
\begin{mycor}\label{cor:obs_equiv_lin}
	For any LTI system satisfying Assumption~\ref{assumption:compact}, the Var-Gram~\eqref{eq:prop_obs_gram} reduces to the linear observability Gramian $\m{W}_{o}^{l}({\m x}_0)$.
\end{mycor}
\begin{proof}
	For an LTI system with linear measurement mapping $\tfrac{\partial{\m{h}(\m{x}_{k})}}{\partial\m{x}_k} = \m{C}$ $\forall \; k \in \{0,1,\cdots,\mr{N}-1\}$ and under the action of the chain rule as discussed in Remark~\ref{rmk:Phi}, the variational measurement equation~\eqref{eq:var_measure} reduces to the following
	\begin{equation}\label{eq:cor1}
		\m{\Psi}^{k}_{0} = \dfrac{\partial{\m{h}(\m{x}_{k})}}{\partial\m{x}_k}\m{\Phi}^{k}_{0}= {\m{C}}\prod^{i=k}_{1}\m{\Phi}^{i}_{i-1},
	\end{equation}
	where $\m{\Phi}^{k}_{0}$ is equivalent to $\m{A}^{k}$ for all $k$, given that $\m{A}^{k}$ is invariant along the system trajectory. It follows that for an LTI system, the variational observability matrix reduces to the linear observability matrix as follows
	$$\m{\Psi}= \m{O}_{l} = \bmat{{\m{C}}^{\top},\;{\m{C}\m{A}}^{\top},\;\cdots\;,\;{\m{C}\m{A}^{\mr{N}-1}}^{\top}}^{\top}; \;\; \text{refer to~\eqref{eq:prop_obs_matrix}}.$$
	While noting that the multiplication of any matrix-valued vector with its transpose is equivalent to its matrix dot-product, the following holds true
	\begin{equation}\label{eq:cor2}
		\m{V}_{o}({\m x}_0) = \m{\Psi}^{\top}\m{\Psi}  =  \langle{\m{\Psi}},{\m{\Psi}}\rangle \equiv \langle{\m{O}}_{l},{\m{O}}_{l}\rangle =  \m{W}_{o}^{l}({\m x}_0).
	\end{equation}
	The proof is therefore complete.
\end{proof}

Having established the above relation to the linear observability Gramian, we now relate the introduced Var-Gram~\eqref{eq:prop_obs_gram} to the Empr-Gram~\eqref{eq:Emp0bsGram} for nonlinear dynamical systems.  
For succinctness in the exposition of the proof, the following theorem is considered under a linear output mapping $\m{h}(\m{x}_k) = \m{C}\m{x}_k$. The linear measurement model is suitable for analyzing sensors that directly measure the state variables at the nodes where they are placed. We note that this choice of measurement model does not restrict the generality of the proofs. The results thus extend to general nonlinear output functions; however, the full proof is beyond the scope of this paper. Readers are referred to Appendix~\ref{apndx:generaloutput} for a discussion.
\begin{theorem}\label{theo:Equivelence}
	Consider the discrete-time nonlinear mapping function $\tilde{\m{f}}(\m x_k)$ and $\m{h}(\m x_k) = \m{C}\m{x}_k$ under the condition of differentiability and Assumption~\ref{assumption:compact}. Then, the Empr-Gram~\eqref{eq:Emp0bsGram} is equivalent to the proposed Var-Gram~\eqref{eq:prop_obs_gram} for any initial condition $\m{x}_{0} \in \mc{\m{X}}_{0}$.
\end{theorem}
\begin{proof}
	See Appendix~\ref{apdx:proof_theo2}.
\end{proof}

Theorem~\ref{theo:Equivelence} demonstrates that the Var-Gram, based on the infinitesimal variational system~\eqref{eq:DiscVar}, is a model-based formulation equivalent to the impulse response-based Empr-Gram in depicting local observability.
It is worthwhile to note that the Empr-Gram, although computed from simple algebraic equations, is considered computationally expensive~\cite{Haber2018}, requiring $2n_x$ impulse response simulations of~\eqref{eq:model_DT} to construct~\cite{Mesbahi2021, Kawano2021}.
Alternatively, the Var-Gram, which relies on variational dynamics, is computed along one local trajectory for any $\m{x}_0 \in \mathcal{X}_0$. A comparison of computational effort is illustrated in Section~\ref{sec:casestudies}.
It is also important to note that the Var-Gram, which is directly computed from the variational dynamics, depicts local variations along the system trajectory and inherently scales the internal state variables to account for such variations. Nevertheless, this is not true for the Empr-Gram, which requires studying the state variables to properly scale internal states relative to the size of the eigenvalues of the Empr-Gram~\cite{Krener2009}. We note here that the Var-Gram requires that the nonlinear mapping functions $\m{f}(\cdot)$ and $\m{h}(\cdot)$ to be smooth and at least twice differentiable, while the Empr-Gram requires the system to be only numerically integrable.
\vspace{-0.2cm}
\section{Observability Conditions \& LEs}\vspace{-0.1cm}\label{sec:LyapExp}
There exist several observability measures that can be defined based on the rank, smallest eigenvalue, condition number, trace, and determinant of an observability Gramian---see \cite{Pasqualetti2014,Summers2016} and references therein. In this section, we show evidence of a connections between an observability measure related to the Var-Gram and LEs. 

LEs measure the exponential rate of convergence and divergence of nearby orbits of an attractor in the state-space~\cite{Pikovsky2017}.
The exponents provide a characteristic spectrum that offers a basis for stability notions introduced in Lyapunov's work on the problem of stability of motions~\cite{AleksandrMikhailovichLyapunov1892}. The maximal LE (MLE) indicates the exponential divergence or convergence about specific trajectory.  In such context, $\m \delta{\m{x}_0}$ is considered an infinitesimal perturbation $\eps>0$ to initial conditions $\m{x}_0$ and its exponential decay or growth is $\m{\delta}\m{x}_k$. Finite-time LEs quantify the evolution of such infinitesimal perturbations over a fixed finite-time horizon~\cite{Krishna2023}, and their computation is well-established in the literature~\cite{Krishna2023,Pikovsky2017}. For nonlinear systems, we compute LEs for horizon $\mr{N}$, using the infinitesimal discrete-time variational representation~\eqref{eq:DiscVar} of the nonlinear dynamics in~\eqref{eq:flow_CT}; see~\cite{Shadden2005a} for a complete derivation of~\eqref{eq:DiscVar} in the field of chaos and ergodicity.

Given a discrete-time flow map for a given trajectory $\m{\phi}^{k}_{0} \in \mathcal{\m{X}}$ starting from initial point $\m{x}_{0} \in \mathcal{\m{X}}_0$, the finite-time maximal LE can be computed, over a finite horizon $k = \mr{N}$, as follows
\begin{equation}\label{eq:LyapExp2}
	\hspace{-0.1cm}\m{\lambda}_{\small \mr{MLE}} \hspace{-0.05cm}:= \hspace{-0.05cm}
	\frac{1}{{k}}\log \left( \dfrac{\norm{\m{\delta}\m{x}_{k}}}{\norm{\m{\delta}\m{x}_0}}\right)
	\hspace{-0.1cm} =\hspace{-0.1cm}
	\frac{1}{2k}\log \left( \norm{\m{\Phi}_{0}^{k}}^{2}\right)\hspace{-0.1cm},
\end{equation}
where the norms $\norm{\m{\delta}\m{x}_k}$ and $\norm{\m{\delta}\m{x}_0}$ represent the average magnitude of the perturbation on $\m{x}_{k}$ and $\m{x}_0$. The induced norm $\norm{\m{\Phi}_{0}^{k}}^{2} := \lambda_{\mr{max}}({\m{\Phi}_{0}^{k}}^{\top} \m{\Phi}_{0}^{k}) $ represents the deformation of an infinitesimal volume, where ${\m{\Phi}_{0}^{k}}^{\top} \m{\Phi}_{0}^{k}$ is the Cauchy-Green deformation matrix---see~\cite{Nolan2020,Krishna2023}. The spectrum of finite-time LEs, for finite-time observation horizon $k=\mr{N}$, is then computed by taking the eigenvalues of the following matrix
\begin{equation}\label{eq:LyapExp3}
	\m{\Lambda}_{L}(\m{x}_0) :=
	\frac{1}{2k}\log \left( {\m{\Phi}_{0}^{k}}^{\top}\m{\Phi}_{0}^{k}\right)\hspace{-0.1cm},
\end{equation}
\noindent where $\m{\lambda}_{L}(\m{x}_0):= \m{\lambda}_{L}= \diag{\lambda(\m{\Lambda}_{L}(\m{x}_0) )} \in \mathbb{R}^{n_x \times n_x} $ is a diagonal matrix of finite-time LEs. The above method for computing the LEs is termed the Oseledec approach; refer to~\cite[Section 2]{Ershov1998}.
\begin{myrem}\label{rmk:regular}
	The regularity of the dynamical system~\eqref{eq:model_DT} is a mild condition. The existence of the full spectrum of LEs is well-established as a result of the multiplicative ergodic theorem (Oseledets Theorem); see~\cite[Section 10.1]{Barreira1998}. 
\end{myrem}

It follows that for discrete-time systems, Oseledec’s splitting fully decomposes Lyapunov vectors and therefore guarantees the existence of a full Lyapunov spectrum of exponents~\cite{Manneville1990, Barreira1998,Frank2018,Tranninger2020,Martini2022}.
This result eliminates the technical challenges of computing the LEs while requiring the verification of system regularity, defined by the existence and smoothness of variational mapping functions $\m{\Phi_{0}^{k}}$.

Before establishing the relation between LEs and an observability measure based on the Var-Gram, we note that LEs are invariant under nonzero scalar scaling and satisfy a max-type inequality under superpositions (where the LEs of a superposed perturbation is dominated by the largest growth rate among its components); see~\cite[Theorem~2.1.2]{Barreira1998}. These properties apply along variational system trajectories given by $\m{\delta x}_k = \m{\Phi}_0^k(\m x_0)\m{\delta x}_0$ for arbitrary $\m{\delta x}_0 \in \mathbb{R}^{n_x}$. As such, the LEs measuring perturbation growth along the variational system trajectory induced by $\m{\Phi}_0^k$ are invariant under nonzero scalar scaling of the perturbation direction and are unaffected by measurement model scaling.

In the following, we illustrate that the $\log\det$ of the proposed Var-Gram~\eqref{eq:prop_obs_gram} is related to the system's LEs~\eqref{eq:LyapExp3}. The following theorem considers a linear measurement mapping function given by $\m{h}(\m{x}_k) = \m{C}\m{x}_k$ that directly measures all state variables, i.e., $\m{C} = \m{I}_{n_x}$.
\begin{theorem}\label{theo:logdet}
Let the Var-Gram be defined as~\eqref{eq:prop_obs_gram} and the LEs be defined as~\eqref{eq:LyapExp3}. The $\log\det$ of the Var-Gram, for $\m{h}(\m x_k) = \m{C}\m{x}_k$, has an underlying connection to the Lyapunov spectrum of exponents according to the following
\begin{equation}\label{eq:Lyap_connection}
\log\det(\m{V}_{o}({\m x}_0)) \equiv \alpha \sum_{i=1}^{n_x}\lambda_{L,i},
\end{equation}
where $\lambda_{L,i}$ are the LEs, i.e., the eigenvalues of diagonal matrix $\m{\lambda}_{L}$, and the constant $\alpha = 2\mr{N}$.
\end{theorem}
\begin{proof}
For the proof of Theorem~\ref{theo:logdet}, see Appendix~\ref{apdx:proof_theo2}.
\end{proof}

Having provided the above relation between the proposed Var-Gram measure and LEs, the following theorem illustrates a local observability condition for discrete-time nonlinear systems without inputs.

\begin{theorem}\label{theo:spectral-radius}
For any discrete-time nonlinear system~\eqref{eq:model_DT} satisfying Assumption~\ref{assumption:compact} and regularity, the system is uniformly observable around $\m{x}_0 \in \mathcal{\m{X}}_0$ if, for a finite-time horizon $k=\mr{N} \in \mathbb{N}$ and for $\m V_{o}(\m{x}_0)\succeq0$ it holds that
\begin{equation}\label{eq:spectral}
\boxed{\frac{1}{2k} \rho(\m{V}_o) = \frac{1}{2k} \lambda_{\max}\!\big(\m V_o(\m x_0)\big) < 1,}
\end{equation}
where $\rho(\m V_{o})\hspace{-0.1cm}:=\hspace{-0.1cm}\rho\big(\m V_o(\m x_0)\big)$ is the spectral radius of the Var-Gram.
\end{theorem}
\begin{proof}\label{proof:spect}
The proof follows from Theorem~\ref{theo:logdet}, where the Gramian~\eqref{eq:prop_obs_gram} is shown to be equivalent to the LEs according to~\eqref{eq:Lyap_connection}. The exponents denote the exponential asymptotic stability around an ellipsoid $\m\delta\m{x}_0$ along the trajectory. For uniform observability to hold, the eigenvalues of $\m{V}_{o}$ must remain bounded and non-degenerate over time. Now, noting that if the LEs are positive $\lambda_{L,i}>0$, the vectors formed from the intersection of $\m{\Phi}_{0}^{k} \cap \m{\Phi}_{0}^{1} = 0$; refer to~\cite[Th.2.4]{Barabanov2005}. Note that $\m{\Phi}_{0}^{k}$ is iteratively calculated based on $\m{\Phi}_{0}^{1}$ (Remark~\ref{rmk:Phi}), it follows that the Var-Gram computation results in unbounded divergence~\cite[Th.2.4]{Barabanov2005}. This indicates rank deficiency in $\m{V}_o$ and thus a loss of information along the trajectory from time index $0$ to $k$. Based on the equivalence relation provided in Theorem~\ref{theo:logdet}, the LEs are to be negative. It follows that the maximal eigenvalue must be less than one, i.e., the spectral radius satisfies $\tfrac{1}{2\mr{N}}\rho(\m{V}_{o}) < 1$. This concludes the proof.
\end{proof}

We note that the above condition addresses uniform observability (Definition~\ref{def:uniformobs}). The following corollary provides necessary and sufficient conditions for the local observability of nonlinear system~\eqref{eq:model_DT} represented in discrete-time variational form~\eqref{eq:DiscVar}.
\begin{mycor}\label{cor:JointSpect}
Let $\m{V}_{o}(\cdot) = {\m{\Psi}(\m{x}_{0})}^{\top}\m{\Psi}(\m{x}_{0})$ denote the Var-Gram computed for the variational system~\eqref{eq:DiscVar}. The spectral radius satisfies
$\limsup_{k\to\infty} ||\m{\Psi}_{0}^{k}||^{\tfrac{1}{k}} \leq 1$
 if and only if there exists a constant $K$ such that $||\m{\Psi}_{k-1}^{k}\;
\m{\Psi}_{k-2}^{k-1}\;
\cdots 	\;\m{\Psi}_{0}^{1}|| \leq K\; \forall\; k\in\mathbb N$.
\end{mycor}\begin{proof}\label{proof:IFF}
The equivalence between the boundedness of a matrix product norm and the spectral radius condition is established in~\cite{Rota1960}. The proof follows from the subadditivity property and the equivalence of norms as presented in~\cite[Lemma 1]{Rota1960}.
\end{proof}	

Theorem~\ref{theo:logdet} shows that the $\log\det$ of the var-Gram is equivalent to computing the LEs of a system. The latter being computationally tractable given that it can be computed through data-driven approaches~\cite{Pikovsky2017}. This equivalence provides data-driven prospects for observability quantification. Theorem~\ref{theo:spectral-radius} provides a uniform observability condition around $\m{x}_0$ for assessing the observability of discrete-time nonlinear systems. Such observability condition can be used in the context of sensor selection and state-estimation. 

The condition presented in Corollary~\ref{cor:JointSpect} holds true for discrete nonlinear systems under Assumption~\ref{assumption:compact} and with the regularity of the system~\eqref{eq:DiscVar}; see Remark~\ref{rmk:regular}. The condition of regularity established by the multiplicative ergodic theorem ensures the boundedness of $\m{\Phi}_{0}^{k}$ and $\m{\Psi}_{0}^{k} \;\; \forall \;\; k \; \in\; \{0,\;1,\cdots,\;\mr{N}-1\}$ and thus the existence of a finite joint spectral radius (Corollary~\ref{cor:JointSpect}). The spectral radius limit in turn presents necessary conditions for local uniform observability (Definition~\ref{def:uniformobs}).
In the subsequent section, we introduce the observability-based SNS framework.

\vspace{-0.4cm}
\section{Application of Observability-Based SNS in Nonlinear Networks}\vspace{-0.1cm}\label{sec:SNS1}
The interdependence between internal state variables allows for the reconstruction of state variables by measuring a subset of the measured outputs. This formulates the basis of observability-based SNS for dynamical systems. While myriad methods exist for addressing the SNS problem in linear systems, approaches for nonlinear systems are less developed. One approach is to pose the problem as a constrained set maximization problem~\cite{Summers2016}; see also~\cite{Joshi2009,Manohar2022} and references therein for alternative methods.

To that end, we define the \textit{set function} $\mathcal{O}{(\mathcal{S})}: 2^{\mathcal{V}}\rightarrow \mbb{R}$ with $\mathcal{V} := \{ i\in\mbb{N}\,|\,0 < i \leq n_y\}$. The set $\mathcal{V}$ denotes all the possible sets of sensor locations combinations and set $\mathcal{S}$ represents a set of sensor combinations such that $\mathcal{S}\subseteq \mathcal{V}$. The SNS set optimization problem can be written as
\begin{equation}\hspace*{-0.3cm}\label{eq:OSP}
{(\mathbf{P1})}\;\;	\mathcal{O}^*({\mathcal{S}}) := \maximize_{\mathcal{S}\subseteq\mathcal{V}}\;\; \mathcal{O}(\mathcal{S}),\;\; \subjectto\; \; \abs{\mathcal{S}} = r.\hspace{-0.2cm}
\end{equation}

In the context of applications to optimal SNS, solving $\mathbf{P1}$ refers to finding the best sensor configuration $\mathcal{S}$ containing $r$ sensors, such that an observability-based measure $\mathcal{O}(\mathcal{S})$ is maximized. The rationale for posing the SNS problem as a set optimization problem is due to the modular and submodular properties of the observability measures, which enable scalable solutions. The following are definitions of modular and submodular set functions.
\begin{mydef}[modularity~\cite{Lovasz1983}]\label{def:modularity}
\label{def:modular_submodular}
A set function $\mathcal{O}: 2^{\mathcal{V}}\rightarrow \mbb{R}$ is said to be modular if and only if for any $\mathcal{S}\subseteq\mathcal{V}$ and weight function $w:\mathcal{V}\rightarrow \mbb{R}$, it holds that $\mathcal{O}(\mathcal{S}) = w(\emptyset) + \sum_{s\in\mathcal{S}} w(s)$. 
\end{mydef}
\begin{mydef}[submodularity~\cite{Lovasz1983}]\label{def:sub}
A set function $\mathcal{O}: 2^{\mathcal{V}}\rightarrow \mbb{R}$ is said to be submodular if and only if for any $\mathcal{A},\mathcal{B}\subseteq\mathcal{V}$ given that $\mathcal{A}\subseteq\mathcal{B}$, it holds that for all $s\notin\mathcal{B}$
\begin{align}
\mathcal{O}(\mathcal{A}\cup\{s\}) - \mathcal{O}(\mathcal{A})\geq 	\mathcal{O}(\mathcal{B}\cup\{s\}) -  \mathcal{O}(\mathcal{B}). \label{eq:submodular_def}
\end{align}
\end{mydef}

We note that for the observability measure function $\mc{O}(\mc{S}) = \log\det(\cdot)$, the SNS problem $\mathbf{P1}$ based on the proposed Var-Gram is submodular and monotone increasing\footnote{Let $\mc{O}: 2^{\mathcal{V}}\rightarrow \mbb{R}$ denote a set function. For any $\mathcal{A},\mathcal{B}\subseteq\mathcal{V}$, the set function is monotone increasing if, for $\mathcal{A}\subseteq\mathcal{B}$ the following is true, $	\mc{O}(\mathcal{B})\geq \mc{O}(\mathcal{A})$.}. This is analogous to the case for linear Gramians, as shown in~\cite{Summers2016}. Note that while several observability measures exhibit modular or submodular properties, such as the $\mr{trace}$ and $\mr{rank}$, other measures, including the condition number and the $\mr{trace}$ of the inverse Gramian, do not exhibit submodular properties. Recall that the variational form~\eqref{eq:DiscVar} of a discrete-time system represents the system flow along the tangent space which is a linear space. We choose the observability measure ($\log\det$) due to its underlying connection to LEs; see~Theorem~\ref{theo:logdet}. Before introducing Theorem~\ref{prs:VarGramModular_prop} and its Corollary~\ref{prs:Trace-Logdet_prop}, which demonstrate the modularity of the Var-Gram and the submodularity of the $\log\det$ observability measure, we define parameterization matrix $\m\Gamma := \mathrm{diag}\{{\gamma_j}\}_{{j}=1}^{n_y}\hspace{-0.05cm}\in \mbb{R}^{n_y\times n_y}$ as the matrix that determines the allocation of the sensors. A node $j$ is equipped with a sensor if $\gamma_j= 1$; otherwise, $\gamma_j$ is set to $0$. We also define parameterization vector $\m \gamma$ that represents the sensor selection, i.e., a column vector $\m \gamma := \{\gamma_j\}_{j=1}^{n_y}$. The measurement mapping function $\m{h}(\m{x}_{k})$ can then be defined as $\m{h}(\m{x}_{k}):= \m \Gamma \m C \m{x}_{k}$. Note that for $\mathbf{P1}$, the variable $\m{\Gamma}$ is encoded in the set $\mathcal{S}$, such that each sensor node corresponds to a value $\gamma_j$ attributed to the set $\mathcal{S}$ at location $j$.
The following theorem demonstrates that the proposed Var-Gram under the context of SNS is a modular set function.
\begin{theorem}\label{prs:VarGramModular_prop}
The Var-Gram $\m{V}_{o}(\mathcal{S},\m{x}_{0}):=\m{V}_{o}(\mathcal{S}) \in \mathbb{R}^{n_x\times n_x}$ defined by
\begin{equation}\label{eq:VarGramSet}
\m{V}_{o}(\mathcal{S}) =  \m{\Psi}(\mathcal{S})^{\top}\m{\Psi}(\mathcal{S}),
\end{equation}
for $\mathcal{S}\subseteq\mathcal{V}$ is a modular set function under parameterization $\m{\gamma}$.
\end{theorem}
\begin{proof}
For the proof of Theorem~\ref{prs:VarGramModular_prop}, see Appendix~\ref{apdx:proof_theo2}.
\end{proof}

Notice that from Definition~\ref{def:modularity}, the Var-Gram is considered a linear mapping function with respect to the sensor selection parameterization vector $\m{\gamma}$. To that end, in the following proposition, we establish the submodularity of the Var-Gram observability-based measure $\mathcal{O}(\mathcal{S})$.
\begin{mycor}\label{prs:Trace-Logdet_prop}
Let $\mathcal{O}(\mathcal{S}):2^{\mathcal{V}}\rightarrow\hspace{-0.05cm}\mbb{R}$  be a set function defined as
\begin{equation}
\mathcal{O}(\mathcal{S}):=\log\det\left({\m V}_{o}(\mathcal{S})\right),	 \label{eq:logdet_submodular}
\end{equation}
for $\mathcal{S}\subseteq\mathcal{V}$. Then $\mathcal{O}(\mathcal{S})$ is a submodular monotone increasing set function.
\end{mycor}
\begin{proof}\label{proof:submodular-log-det}
	Let $\mc{O}_s: 2^{\mathcal{V}\setminus\{s\}} \rightarrow \mathbb{R}$ denote a derived set function defined as
\begin{align*}
			\mc{O}_s(\mathcal{S}) & =\mr{log}\,\mr{det} \m{V}_o\left({\mathcal{S} \cup\{s\}}\right)-\mr{log}\,\mr{det} \m{V}_{o}(\mathcal{S}),\\
			& =\mr{log}\,\mr{det}\left(\m{V}_{o}(\mathcal{S})+\m{V}_{o}(\{s\})\right)-\mr{log}\,\mr{det} \m{V}_{o}(\mathcal{S}).
	\end{align*}
	
	We first show $\mc{O}_s(\mathcal{S})$ is monotone decreasing for any $s \in \mathcal{V}$. That being said, let $\mathcal{A} \subseteq \mathcal{B}\subseteq \mathcal{V}\setminus\{s\}$, and let $\m{V}_o(\tilde{\m{c}}) = \m{V}_o(\mathcal{A}) +\tilde{\m{c}}\left( \m{V}_o(\mathcal{B})- \m{V}_o(\mathcal{A})\right)$ for $\tilde{\m{c}}\in[0,1]$. Then for 
		\begin{align*}
			\mc{\tilde{O}}_s(\m{V}_o(\tilde{\m{c}}))=\mr{log}\,\mr{det}\left(\m{V}_o(\tilde{\m{c}})+\m{V}_{o}(\mathcal{S})\right)-\mr{log}\,\mr{det}\left(\m{V}_o(\tilde{\m{c}})\right),
		\end{align*}
		we obtain the following
		\begin{align*}\label{eq:proof:cor3}
			\frac{\mr{d}}{\mr{d} \tilde{\m{c}}} \mc{\tilde{O}}_s(\m{V}_o(\tilde{\m{c}}))
				&=
			\mr{trace}\Big[
			\left(\left(\m{V}_o(\tilde{\m{c}})+\m{V}_{o}(\mathcal{S})\right)^{-1}-\m{V}_o(\tilde{\m{c}})^{-1}\right)\\
			&\quad\quad\quad\quad\;
			\left(\m{V}_o(\mathcal{B})-\m{V}_o(\mathcal{A})\right) \Big] \leq 0,
		\end{align*}
		where $
			\left(\left(\m{V}_o(\tilde{\m{c}})+\m{V}_{o}(\mathcal{S})\right)^{-1}-\m{V}_o(\tilde{\m{c}})^{-1}\right) \preceq 0$, and $\left(\m{V}_o(\mathcal{B})-\m{V}_o(\mathcal{A})\right)\succeq 0$, then the above inequality holds. Then, for 
			$\tfrac{\mr{d}}{\mr{d} \tilde{\m{c}}} \mc{\tilde{O}}_s(\m{V}_o(\tilde{\m{c}})) \leq 0$, it follows that $\mc{\tilde{O}}_s(\m{V}_o({1})) \leq \mc{\tilde{O}}_s(\m{V}_o({0}))$, which implies that $\mc{O}_s(\mathcal{A}) \geq \mc{O}_s(\mathcal{B})$. Then, by definition of derived set function, we have $ \mc{O}_s(\mathcal{A})= \mathcal{O}(\mathcal{A}\cup\{s\}) - \mathcal{O}(\mathcal{A})\geq 	\mc{O}_s(\mathcal{B}) = \mathcal{O}(\mathcal{B}\cup\{s\}) -  \mathcal{O}(\mathcal{B})$. We have therefore shown that $\mc{O}(\mathcal{S})$ is submodular and $\mc{O}_s$ is monotone decreasing. Then, by the additive property of $\m{V}_{o}(\mathcal{S})$ (see~\cite{Summers2016}) we have $\mc{O}(\mathcal{S})$ being monotone increasing. The proof is analogous and a corollary to the results in~\cite[Theorem 6]{Summers2016} and~\cite[Lemma 3]{Zhou2019}.
\end{proof}

Theorem~\ref{prs:VarGramModular_prop} provides evidence that the Var-Gram is modular, which in turn implies the submodularity of the $\log\det$ observability measure; see~Corollary~\ref{prs:Trace-Logdet_prop}. This submodularity is a direct consequence of the modularity of the proposed Gramian. Note that the $\log\det$ of the Var-Gram can have zero eigenvalues and still remain submodular and monotone increasing. In the study~\cite{Zhou2019}, the observability measures are based on the Lie derivative matrix $\m{O}_{l}$. However, the submodular properties discussed above hold if and only if $\m{O}_{l}$ is full rank. In contrast, when considering the Var-Gram, there is no such restriction. The submodularity of the $\log\det$ still holds in rank-deficient situations. These conditions can occur when not enough sensing nodes are chosen, resulting in a system that is not yet observable.

Having a submodular set function allows one to exploit computationally tractable algorithms to solve $\mathbf{P1}$. Typically, a greedy algorithm, with a running time complexity of $\mathcal{O}(|\mathcal{S}||\mathcal{V}|)$, is employed to solve the submodular problem---see~\cite[Algorithm 1]{Kazma2023f}. Note that for a set function $\mc{O}: 2^{\mathcal{V}} \rightarrow \mathbb{R}$ that is submodular and monotone increasing, the algorithm has the following theoretical performance guarantee
\begin{align*}
	\mc{O}^*_{\mathcal{S}} -\mc{O}(\emptyset) \geq \left(1-\tfrac{1}{e}\right)\left(\mc{O}^*-\mc{O}(\emptyset)\right), \quad \text{with}\; \mc{O}(\emptyset) =0,	\end{align*}
where $\mc{O}^*$ is the optimal solution of the SNS problem and $\mc{O}^*_{\mathcal{S}}$ is the solution obtained from the greedy algorithm. Given that $(1 - 1/e) \approx 0.63$, the solution $\mc{O}_{\mathcal{S}}^*$ is guaranteed to be at least $63\%$ of the optimal value $\mc{O}^*$. However, in practice, for submodular set maximization problems, an accuracy of $99\%$ can be achieved~\cite{Summers2016}. In the section below, we demonstrate the applicability of the Var-Gram to SNS in nonlinear networks.

\vspace{-0.4cm}
\section{Numerical Case Studies}\vspace{-0.15cm}\label{sec:casestudies}
To demonstrate the effectiveness of the proposed method, we investigate the following research questions.
\begin{itemize}[]
\item $(\mathrm{Q}1)$ Theorem~\ref{theo:Equivelence} establishes the equivalence between the Var-Gram and Empr-Gram.  Does the Var-Gram dynamically depict the intrinsic relations between the state variables, and is it computationally more tractable than the Empr-Gram? \item $(\mathrm{Q}2)$ Does the established observability condition presented in~Theorem~\ref{theo:spectral-radius} hold true for the system under study?
\item  $(\mathrm{Q}3)$ Having provided evidence regarding the modularity of the Var-Gram, does solving $\mathbf{P1}$ result in optimal sensor node selections? Is the proposed optimal SNS problem scalable for larger nonlinear networks?
\end{itemize}
\begin{table}[t]
\fontsize{9}{9}\selectfont
\centering 
	\caption{Computational time of the observability Gramians. }
\label{tab:ODE_disc}
	\vspace{-0.2cm}
	\renewcommand{\arraystretch}{1.3}
	\begin{tabular}{l|l|l|l|l}
		\midrule \hline
		\multirow{2}{*}{Network} & \multicolumn{1}{c|}{Perturbation}                             & \multicolumn{2}{c}{Computational Time $(\mr{s})$}                           \\ \cline{2-4} 
		& \multicolumn{1}{c|}{$\alpha_{L}$} & \multicolumn{1}{c|}{Var-Gram} & \multicolumn{1}{c}{Empr-Gram} \\ \hline
		\multirow{2}{*}{$\mr{H}_2\mr{O}_2$} 
		& \multicolumn{1}{c|}{$20\%$} & \multicolumn{1}{c|}{$0.0043$} & \multicolumn{1}{c}{$7.38$} \\ \cline{2-4} 
		& \multicolumn{1}{c|}{$30\%$} & \multicolumn{1}{c|}{$0.0032$} & \multicolumn{1}{c}{$8.53$}  \\ \cline{2-4} \hline
		\multirow{2}{*}{$\mr{GRI}30$} 
		& \multicolumn{1}{c|}{$20\%$} & \multicolumn{1}{c|}{$0.489$} & \multicolumn{1}{c}{$115.05$}  \\ \cline{2-4} 
		& \multicolumn{1}{c|}{$30\%$} & \multicolumn{1}{c|}{$0.467$} & \multicolumn{1}{c}{$112.06$} \\ \cline{2-4} 
		\toprule \bottomrule
	\end{tabular}
	\vspace{-0.1cm}
\end{table}
\setlength{\textfloatsep}{3pt}

We consider a general nonlinear combustion reaction network~\cite{smirnov2014modeling} with a state-space formulation of the following form
\begin{equation}\label{eq:combustion}
\dot{\m {x}}(t)=\Theta \boldsymbol{\psi}\left(\m {x}(t)\right),
\end{equation}
where $\boldsymbol{\psi}\left(\m {x}\right)=[\psi_{1}\left(\m {x} \right),\psi_{2}\left(\m {x} \right),\ldots, \psi_{N_{r}}\left(\m {x} \right)]^{\top}$, such that $\psi_{j}$ for $j\in\{1,\;2,\;\ldots\;, \; N_{r}\}$ are the polynomial functions of concentrations. State vector $\m {x}=[x^{1},\; x^{2},\;\cdots\;,\; x^{n_{x}}]$, represents the concentrations of $n_x$ chemical species. The constant matrix $\Theta \in \mathbb{R}^{n_{x}\times N_{r}}$ is defined entrywise as $[\theta_{ij}] := [w_{ji}-q_{ji}]$, where $q_{ji}$ and $w_{ji}$ are the stoichiometric coefficients, with $i\in\{1,\ldots,n_x\}$ and $j\in\{1,\ldots,N_r\}$.
The number of chemical reactions is denoted by $N_{r}$ and the list of reactions can be written as $\sum_{i=1}^{n_{x}}q_{ji} \mathcal{R}_{i} \rightleftarrows  \sum_{i=1}^{n_{x}} w_{ji}\mathcal{R}_{i}, \; j \in \{1,\;2\;,\;\cdots\;, \; N_{r}\}$, where $\mathcal{R}_{i}$, $i\in \{1,\;2,\;\cdots\;,\; n_{x}\}$ are the chemical species. 

\begin{figure}[t]
\centering
\includegraphics[keepaspectratio=true,scale=0.96]{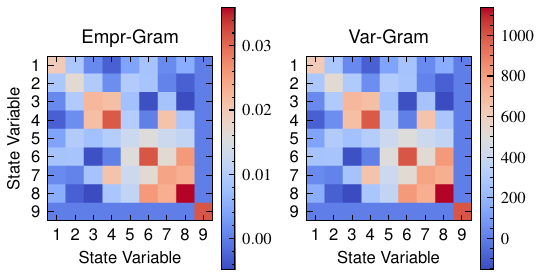}
\vspace{-0.7cm}
\caption{Mapping of the Empr-Gram (left) and Var-Gram (right). The square colors indicate the strength and direction of the relations between the variables. Color contrast represents the strength of the relation, while the color itself indicates the sign of the interrelations: red for positive and blue for negative correlation between state variables. Notice that along the diagonal, values are always positive. The two Gramians are shown to be equivalent (refer to Theorem~\ref{theo:Equivelence}).}
	\label{fig:Gram}
	\vspace{-0.4cm}
\end{figure}

\begin{figure}[t]
	\centering
	\vspace{0.1cm}
	\includegraphics[keepaspectratio=true,scale=0.55]{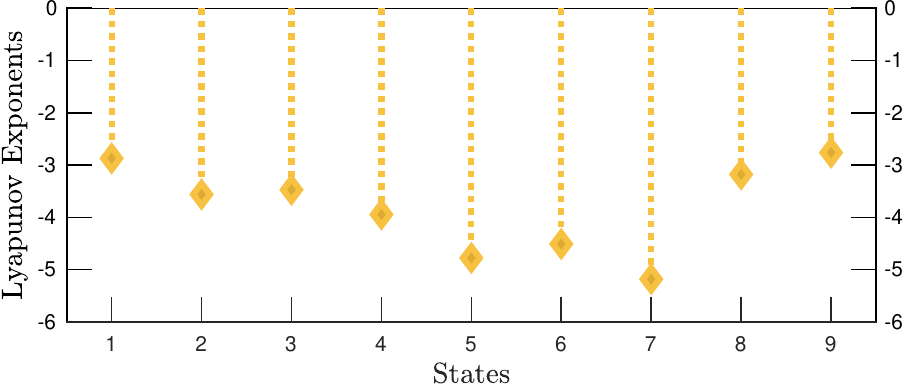}
\vspace{-0.2cm}
	\caption{LEs computed for the $\mr{H}_2\mr{O}_2$ reaction network. As indicated by the direction of the dotted lines and the location of the diamond markers, the exponents $\lambda_{L,i}$ are all negative and therefore satisfy the observability condition in Theorem~\ref{theo:spectral-radius}.}\label{fig:LyapExp}
	\vspace{-0.1cm}
\end{figure}

\begin{figure*}[t]
	\centering
	\subfloat{\includegraphics[keepaspectratio=true,scale=0.96]{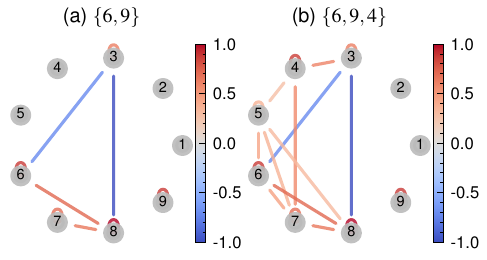}}	\subfloat{\includegraphics[keepaspectratio=true,scale=0.96]{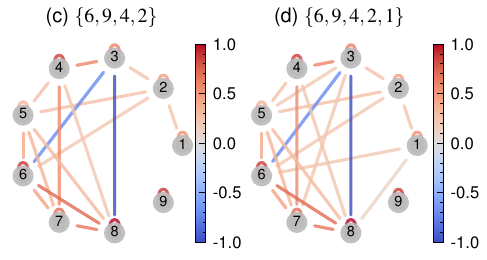}}	
	\vspace{-0.2cm}
	\caption{Optimally selected subset of sensor nodes for the $\mr{H}_2\mr{O}_2$ combustion network under sensor ratios $r=[2,3,4,5]/9$. Sub-figures (a)-(d) depict the sequential order of sensor selections and the corresponding nodes observed. The normalized observability relation between state variables is indicated by the edge color contrast for each sensor ratio. The color intensity represents the strength of observability.}
\label{fig:allocations}
	\vspace{-0.3cm}
\end{figure*}

We study two combustion reaction networks: $(N1)$ an $\mr{H}_2\mr{O}_2$ network that has $N_r = 27$ reactions and $n_x=9$ chemical species and $(N2)$ a $\mr{GRI30}$ network that has $N_r = 325$ reactions and $n_x=53$ chemical species. For specifics regarding system parameters and definitions, we refer the readers to our previous work~\cite[Section $\mr{V}$]{Kazma2023f}. The discretization constant is $T=1\cdot 10^{-12}$ and observation window of $\mr{N} = 1000$ is chosen. The choice of discretization constant is a result of analyzing the system's initial condition response. For the computation of Empr-Gram, the constant $\eps$ is chosen as $\eps = 10^{-4}$. The simulations are performed using the Cantera toolbox~\cite{Cantera} in MATLAB using a MacBook Pro with an Apple M1 Pro chip, a 10-core CPU, and 16 GB RAM. The source code for the case studies can be found in~\cite{KazmaVarGram}.
\vspace{-0.1cm}
\subsection{Observability of a combustion reaction network}\vspace{-0.1cm}\label{sec:case_combustion}
We analyze the proposed Gramian and present a comparison to the Empr-Gram.
The comparison between the studied Gramians for $\mr{H}_2\mr{O}_2$ network is depicted in Fig.~\ref{fig:Gram}. For brevity, we did not include the observability matrix for $\mr{GRI}30$ network. It is clear that the interrelation between the variables is equivalent for the two Gramian formulations. The depicted equivalence is a consequence of Theorem~\ref{theo:Equivelence}. However, as mentioned earlier~\cite{Krener2009}, the Empr-Gram requires proper scaling of the state variables, which can lead to small eigenvalues. This is evident by the strength or amplitude of the state interrelations which is the magnitude of $\{10^{-3}\}$; it is dynamically scaled to a magnitude of $\{10^{3}\}$ in the Var-Gram. This is important since under a dynamic system response, any small perturbations to initial conditions can appear naturally in the Var-Gram formulation.

Having numerically illustrated the equivalence mentioned above, we now present numerical evidence demonstrating the computational efficiency of the proposed method as compared to the Empr-Gram. The computational time required to compute the Var-Gram and the Empr-Gram for both combustion networks $N1$ and $N2$, under different perturbations $\alpha_{L}$ to initial conditions $\m{x}_{0}$, is presented in Tab.~\ref{tab:ODE_disc}. Perturbation term $\alpha_{L}$ is applied to simulate the system under different transient conditions. It is clear that the Var-Gram is more amenable for scaling for larger nonlinear networks due to its lower computational cost.
The aforementioned observations answer the questions in~$\mr{Q}1$. 

\begin{figure}[t]
	\vspace{-0.2cm}
	\centering
\includegraphics[keepaspectratio=true,scale=0.96]{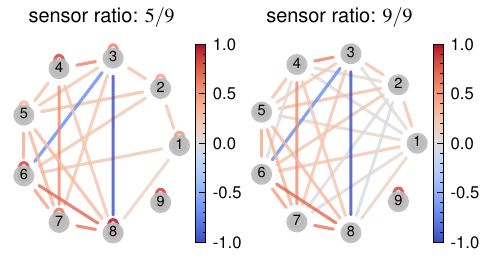}
	\vspace{-0.4cm}
	\caption{Normalized observability relation between state variables based on $(5/9\approx0.7)$ sensor node ratio (left) and full node sensor ratio (right) for the sensed nodes within the nonlinear combustion network. For left figure, sensed nodes are optimally chosen as $\{1,\;2\;,4\;,6\;,9\}$.}\label{fig:OSP}
	\vspace{-0.2cm}
\end{figure}

Fig.~\ref{fig:LyapExp} depicts the value of the LEs for the $\mr{H}_2\mr{O}_2$ combustion network. The LEs are all negative and satisfy the observability condition presented in Theorem~\ref{theo:spectral-radius}.
This is true as a consequence of Theorem~\ref{theo:logdet}, since for any $\lambda_{i}$ value less than $1$, the LEs $\mr{log}{\lambda_{i}} = \lambda_{L,i}<0$, i.e., are negative. This verifies that the system is observable when considering $|\mathcal{S}| = n_x$, i.e., the system has sensors employed on all sensor nodes. As such, the condition posed in research question $\mr{Q}2$ holds true for the considered combustion network.

\vspace{-0.4cm}
\subsection{Sensor node selection in nonlinear networks}\vspace{-0.1cm}\label{sec:SNS}
The greedy algorithm~\cite[Algorithm 1]{Kazma2023f} is employed to solve the SNS problem $\mathbf{P1}$ with~\eqref{eq:logdet_submodular} as the submodular objective function. We solve the optimal SNS problem for combustion networks $N1$ and $N2$ using the $\log\det$ measure. The sensor ratios chosen are $[0.22, 0.33, 0.44, 0.56]\times n_x$. For the $\mr{H}_2\mr{O}_2$ network this is equivalent to sensors $r=[2,3,4,5]$. The results for SNS on $N1$ for each of the sensor ratios are depicted in Fig.~\ref{fig:allocations}. The figure depicts the order in which the sensors are allocated when increasing the sensor ratio and the respective normalized internal state relations between the state variables. For instance, when observing nodes $\{6, 9\}$, the state variables at nodes $\{3, 7, 8\}$ can be inferred due to their internal state connections with node $\{6\}$. Notice that node $\{9\}$ has a self-loop thereby indicating that it is a non-interacting chemical species. This means that node $\{9\}$ is only observable when the same node is selected. The optimally selected subset of the sensor nodes is $\{1\;,2\;,4\;,6\;,9\}$ for the $\mr{H}_2\mr{O}_2$ network under a sensor number $r=5$. Notice that node $\{3\}$ is not selected as a sensing node. This is due to its negative correlation with other state variables (indicated by the blue color in Fig.~\ref{fig:allocations}), which results in a lower contribution to overall observability of the system. To illustrate the optimally selected set, Fig.~\ref{fig:OSP} depicts the normalized internal state relations for sensor ratio $5/9$ and that of a full sensor selection ratio. Note that the normalization is based on a min-max normalization process which results in a normalized observability relation ranging between $\{-1,1\}$. The strength of the relationships and interactions among the state variables for both sensor fractions enables internal state connectivity between all state variables, thereby indicating the optimality of the chosen sensor subset with $r = 5$.
Such equivalence illustrates that the sensor nodes selected using the proposed submodular set optimization framework are optimal. Furthermore, this demonstrates that only 5 sensors are required to observe all the chemical species within the $\mr{H}_2\mr{O}_2$ network. This is also evident in Fig.~\ref{fig:EstErr} where it is shown that the estimation error approaches zero for the optimal solution when considering $r=5$ for the $\mr{H}_2\mr{O}_2$ network. Note that the estimation error is related to the observability of the system. Such that, the estimation error is computed as $e=\left\|\m{x}-\hat{\m {x}} \right\|_{2}/\left\|\m {x} \right\|_{2}$, where $\m{x}$ is the true state and $\hat{\m {x}}$ is the estimate obtained by solving a general nonlinear least-squares state estimation problem over the observation horizon $\mr{N}$. Interested readers are referred to~\cite{Haber2018,Kazma2023f} for the formulation of the state estimator problem.

\begin{figure}[t]
\centering
	\includegraphics[keepaspectratio=true,scale=0.52]{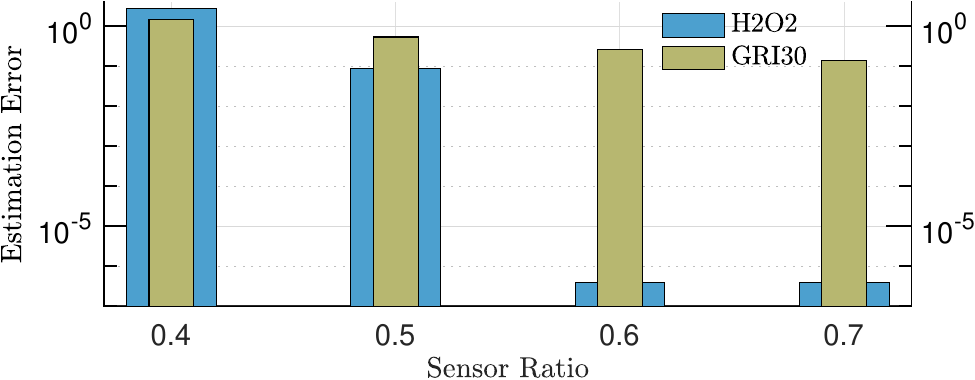}
	\vspace{-0.25cm}
	\caption{State estimation error based on the optimal SNS set $\mc{S}^*$. The performance of the state estimation depends on the degree of system observability.}\label{fig:EstErr}
		\vspace{-0.2cm}
\end{figure}
For the $\mr{H}_2\mr{O}_2$ and $\mr{GRI}30$ networks, we solve $\mathbf{P1}$ for the aforementioned sensor ratios and compute the estimation error. Notice that, for $\mr{GRI}30$ network, the estimation error decreases but does not reach an optimal error value due to a large number of non-interacting species. This indicates that additional sensors are required for better state estimation. We also note that utilizing a greedy approach yields an efficient and scalable solution to the observability-based SNS problem in nonlinear systems as compared with methods that rely on empirical data simulations---as with Empr-Gram framework. Solving $\mathbf{P1}$ for the $\mr{H}_2\mr{O}_2$ under sensor ratio $5/9$ requires $0.414 \; \mr{s}$. The computational effort increases to about $24.797 \; \mr{s}$ for the $\mr{GRI}30$ network when considering an equivalent sensor ratio. This shows that the proposed SNS framework is scalable for larger nonlinear networks. For brevity, we do not introduce or solve the SNS framework based on the Empr-Gram, and therefore, we do not provide a comparison. However, based on the computational time provided in~Tab.~\ref{tab:ODE_disc}, it is inferred that solving $\mathbf{P1}$ requires significantly more effort. The computational efficiency along with the optimality results provided by Fig.~\ref{fig:OSP} and Fig.~\ref{fig:EstErr}, illustrate the modularity of the proposed Var-Gram and its submodular observability measure $\log\det$ for SNS in nonlinear systems and thereby answers question $\mr{Q}3$. On this note, we conclude this section. 

\vspace{-0.1cm}
\section{Summary, Limitation and Future Directions}\vspace{-0.1cm}\label{sec:conclusion}
This paper introduces a new observability Gramian, (Var-Gram), for discrete-time nonlinear systems without inputs. The proposed Var-Gram is based on a discrete-time variational system representation. We prove that the Gramian is equivalent to the Empr-Gram and reduces to the observability Gramian in the linear case. A connection between the Var-Gram and LEs is investigated. We derive a spectral observability limit result based on the Var-Gram. To illustrate the approach, we demonstrate its applicability under the context of observability-based SNS.
The Var-Gram, in its current form, is not devoid of limitations. $(i)$ The method is developed for general nonlinear systems without control input. $(ii)$ The equivalence between the Var-Gram and the LEs is demonstrated for the linear measurement model case. $(iii)$ Connections and comparisons to analytical observability methods based on Lie derivatives are not investigated. The above limitations merit future work. Furthermore, extension of this method for stochastic nonlinear systems with noisy measurements is important; such extension is demonstrated for the Empr-Gram~\cite{Powel2020b}.

\appendices
\vspace{-0.1cm}
 \section{Runge-Kutta Discrete-time Model}\vspace{-0.1cm}\label{apdx:Runge-Kutta}
The nonlinear mapping function $\tilde{\m{f}}(\cdot)$ is defined for the IRK method as 
\begin{equation}
	\tilde{\m{f}}(\m{x}_0):= \tfrac{T}{4}\left(\m f(\m \zeta_{1,k+1})+3\m f(\m \zeta_{2,k+1})\right).
\end{equation}

Vectors $\m \zeta_{1,k+1},\m \zeta_{2,k+1}\in\mbb{R}^{n_x}$ are auxiliary for computing $\m x_{k+1}$ provided that $\m x_{k}$ is given. The IRK method results in the following implicit discrete-time state-space model 
\begin{align}
	\begin{split}
		\m \zeta_{1,k+1} &:= {\m x_{k}} + \tfrac{T}{4}\left(\m f(\m \zeta_{1,k+1})-\m f(\m \zeta_{2,k+1})\right), \\
		\m \zeta_{2,k+1} &:= {\m x_{k}} + \tfrac{T}{12}\left(3\m f(\m \zeta_{1,k+1})+5\m f(\m \zeta_{2,k+1})\right), \\
		\m x_{k+1} &:= \m x_{k}+ {\tfrac{T}{4}\left({\m f(\m \zeta_{1,k+1})+3\m f(\m \zeta_{2,k+1})}\right)}.
	\end{split}\label{eq:IRK_dynamics}
\end{align}

The additional layer that includes calculating auxiliary vectors $\m \zeta_{1,k+1}$ and $\m \zeta_{2,k+1}$ enables an accurate and stable approach for the discretization of a broad class of nonlinear networks.
To evaluate the partial derivative $\tfrac{\partial \tilde{\m f}\left(\m x_k\right)}{\partial \m x_0}$ for use in the discrete-time variational equations~\eqref{eq:DiscVar}, we first need to compute the partial derivative $\tfrac{\partial \m{x}_{k+1}}{\partial \m{x}_{k}}$, which is involved numerically and can be written as follows
\begin{equation}\label{eq:1}
	\begin{aligned}
		\frac{\partial \m{x}_{k+1}}{\partial \m{x}_{k}} := & \m{I}_{n_x}+\left.\frac{T}{4} \frac{\partial \m{f}\left(\m\zeta_{1, k+1}\right)}{\partial \m \zeta_{1, k+1}}\right|_{\m \zeta_{1,k+1}} \times \frac{\partial \m \zeta_{1, k+1}}{\partial \m{x}_{k}}\\ &
		+\left.\frac{3 T}{4} \frac{\partial \m{f}\left(\m \zeta_{2, k+1}\right)}{\partial \m \zeta_{2, k+1}}\right|_{\m \zeta_{2, k+1}}
		\times \frac{\partial \m  \zeta_{2, k+1}}{\partial \m{x}_{k}}.
	\end{aligned}
\end{equation}

Notice that to determine $	\frac{\partial \m{x}_{k+1}}{\partial \m{x}_{k}}$, we need to determine the partial derivatives $\tfrac{\partial \m \zeta_{1, k+1}}{\partial \m{x}_{k}}$ and $\tfrac{\partial \m \zeta_{2, k+1}}{\partial \m{x}_{k}}$. By differentiating~\eqref{eq:IRK_dynamics} with respect to $\m{x}_k$, we obtain
\begin{equation}
\begin{aligned}
	\underbrace{\left[\begin{array}{l}
			\frac{\partial \m \zeta_{1, k+1}}{\partial \m{x}_{k}} \\
			\frac{\partial \m \zeta_{2, k+1}}{\partial \m{x}_{k}}
		\end{array}\right]}_{\m{Q}}:= &\underbrace{\left[\begin{array}{l}
			\m{I}_{n_x} \\
			\m{I}_{n_x}
		\end{array}\right]}_{\m{I}_{2n_x}}+\underbrace{\left[\begin{array}{cc}
			\frac{T}{4} \frac{\partial \m{f}\left(\m\zeta_{1, k+1}\right)}{\partial \m\zeta_{1, k+1}} & -\frac{T}{4} \frac{\partial \m{f}\left(\m\zeta_{2, k+1}\right)}{\partial \m\zeta_{2,k+1}} \\
			\frac{3 T}{12} \frac{\partial \m{f}\left(\m\zeta_{1, k+1}\right)}{\partial \m\zeta_{1, k+1}} & \frac{5 T}{12} \frac{\partial \m{f}\left(\m\zeta_{2, k+1}\right)}{\partial \m\zeta_{2, k+1}}
		\end{array}\right]}_{\m{K}}
	\\& \times \left[\begin{array}{l}
		\frac{\partial \m\zeta_{1, k+1}}{\partial \m{x}_{k}} \\
		\frac{\partial \m\zeta_{2, k+1}}{\partial \m{x}_{k}}
	\end{array}\right].
\end{aligned}
\end{equation}

Assuming that the matrix $\left[\m I_{2 {n_x}}-\m{K}\right]$ is invertible (this is sufficient as a consequence of the implicit function theorem~\cite[Th. 3.3.1]{Krantz2013}), where $\m{K} \in$ $\mathbb{R}^{2 n_x \times 2 n_x}$, we have from the last expression
\begin{equation}\label{eq:2}
	\m{Q}=\left[\m I_{2 {n_x}}-\m{K}\right]^{-1} \m{I}_{2n_x},
\end{equation}
once matrix $\m{Q}$ is computed, its elements can be used to calculate the partial derivatives~\eqref{eq:1}.

Note that the auxiliary vectors $\m{\zeta}_{1, k+1}$ and $\m{\zeta}_{2, k+1}$ also depend on $\m{x}_0$ and can be obtained by simulating the system with initial condition equal to $\m{x}_0$.
Based on the implicit nature of the auxiliary vectors, the vectors are embedded in the computation of 	$\frac{\partial \m{x}_{k+1}}{\partial \m{x}_{k}}$; refer to~\eqref{eq:1}-\eqref{eq:2}. Then, the partial derivative  $\frac{\partial \tilde{\m f}\left(\m x_k\right)}{\partial \m x_0}$ under the action of the chain rule can be written as
\begin{equation}
	\frac{\partial \tilde{\m f}\left(\m x_k\right)}{\partial \m x_0}=\frac{\partial \tilde{\m f}\left( \m x_k\right)}{\partial \m x_k} \frac{\partial \m x_k}{\partial \m x_{k-1}}\frac{\partial \m x_{k-1}}{\partial \m x_{k-2}}\ldots\frac{\partial \m x_1}{\partial \m x_0},
\end{equation}
where $\frac{\partial \tilde{\m f}\left(\m x_k\right)}{\partial \m x_0}$ is essential for computing the variational mapping function $\m{\Phi}_{0}^{k}$ in~\eqref{eq:DiscVar}.
Note that the computation of $\frac{\partial \tilde{\m f}\left(\m x_k\right)}{\partial \m x_0}$ depends on the discretization method. Other methods follow similar derivations.

\vspace{-0.5cm}
\section{Proof of Theorems~\ref{theo:Equivelence},~\ref{theo:logdet}, and~\ref{prs:VarGramModular_prop}}\vspace{-0.1cm}\label{apdx:proof_theo2}
\begin{proofofc}{theo:Equivelence}
	The Empr-Gram~\eqref{eq:Emp0bsGram} can be formulated in differential form by applying the central difference definition of a directional derivative to the impulse response measurement vector $\m{\Delta Y}^{\eps}_{k}$ as~\eqref{eq:prooflem1}, see~\cite{Powel2015,Kawano2021} for additional information. 
\begin{align}\label{eq:prooflem1}
		\m{\Delta Y}^{\eps}_{k}=  \bmat{2\eps\tfrac{\partial\m{y}_{k}}{\partial\m{x}_{0}^{1}},
			\; \cdots \; ,2\eps\tfrac{\partial\m{y}_{k}}{\partial\m{x}_{0}^{n_x}}} \; \in \mbb{R}^{n_y \times n_x},
\end{align}
	where state vector $\m{x}_{0}^{i}\in \mathcal{\m{X}}_0$ is denoted as $\m{x}_{0}^{i} = \m{x}_{0}\pm \eps\m{e}_{i}\;\; \forall \;\; i \in \; \{1,\;2,\;\cdots\;,n_x\}$, thus~\eqref{eq:prooflem1} is equivalent to $2\eps\tfrac{\partial\m{y}_k}{\partial\m{x}_0}$. With that in mind, the Empr-Gram~\eqref{eq:Emp0bsGram} can be written in the following form
	\begin{equation}\label{eq:EmpDiff}
		{\m{W}}_{o}^{\partial}(\m{x}_{0}) := \sum_{k=0}^{\mr{N}-1}
		\tfrac{\partial\m{y}_k}{\partial\m{x}_0}^{\top}
		\tfrac{\partial\m{y}_k}{\partial\m{x}_0} 	\; \in \mbb{R}^{n_x \times n_x}.
\end{equation}
	
	For simplicity of exposition, we consider a linear measurement model $\m{h}(\m{x}_{k}) = {\m{C}}\m{x}$. This does not restrict the proof since we are utilizing the same measurement model for both the Var-Gram and the Empr-Gram. 
It follows from~\eqref{eq:EmpDiff} that for any time index $k$, we have 
	$\tfrac{\partial\m{y}_{k}}{\partial\m{x}_0} \equiv 
	\tfrac{\partial\m{h}(\m{x}_{k})}{\partial \m{x}_{k}}
	\tfrac{\partial \m{x}_{k}}{\partial\m{x}_0} =  {\m{C}}\tfrac{\partial \m{x}_{k}}{\partial\m{x}_0}$. In a similar approach, for the variational Gramian, we obtain $\m{\Psi}_{0}^{k} = {\m{C}}\m{\Phi}_{0}^{k} = {\m{C}}\big(\m{I}_{n_x}  +\tfrac{\partial\tilde{\m{f}}(\m{x}_{k})}{\partial\m{x}_{k}}\big)\tfrac{\partial\m{x}_{k}}{\partial\m{x}_{0}}$.
Note that $\tfrac{\partial \m{x}_{k+1}}{\partial\m{x}_0} = \big(\tfrac{\partial \m{x}_{k}}{\partial\m{x}_{0}} +\tfrac{\partial\tilde{\m{f}}( \m{x}_{k})}{\partial\m{x}_{0}}\big)$ is obtained by taking the partial derivative of~\eqref{eq:RKutta_dynamics_compact} about $\m{x}_0$ and under the action of the chain rule, the partial derivative $\tfrac{\partial \m{x}_{k}}{\partial\m{x}_0}$ becomes a composition mapping similar to $\m{\Psi}_{0}^{k}$; refer to Remark~\ref{rmk:Phi}. Hence, the two Gramians are equivalent; this completes the proof.\end{proofofc}

\begin{proofofc}{theo:logdet}\label{proof:log-det}
Let $\m{h}(\m{x}_{k}) = {\m{C}}\m{x}$. Then the observability matrix can be written as $\m{\Psi} \hspace{-0.05cm}=\hspace{-0.05cm} 
	\left\{\tfrac{\partial{\m{h}(\m{x}_{k})}}{\partial\m{x}_k}\m{\Phi}_{0}^{k}\right\}^{\mr{N}-1}_{k=0}= \left\{\m{C}\m{\Phi}_{0}^{k}\right\}^{\mr{N}-1}_{k=0}$. From~\eqref{eq:prop_obs_gram}, for an observation horizon $\mr{N}$, the Var-Gram is given by
	\begin{align*} \m{V}_{o} &\hspace{-0.1cm}=\hspace{-0.15cm}
		\sum_{k=0}^{\mr{N}-1} \bmat{{\m{C}} \prod^{i=k}_{1} \m{\Phi}^{i}_{i-1}}^{\top}
		\bmat{{\m{C}}\prod^{i=k}_{1}\m{\Phi}^{i}_{i-1}},\hspace{-0.2cm}\\
		&\hspace{-0.1cm}=\hspace{-0.15cm} \sum_{k=0}^{\mr{N}-1}
		\hspace{-0.1cm} \bmat{{{\prod^{i=k}_{1}\m{\Phi}^{i}_{i-1}}^{\hspace{-0.2cm}\top}\m{C}^{\top}}}
		\hspace{-0.15cm}
		\bmat{{\m{C}}\prod^{i=k}_{1}\m{\Phi}^{i}_{i-1}}
		\hspace{-0.1cm}=\hspace{-0.15cm} \sum_{k=0}^{\mr{N}-1}{\hspace{-0.05cm}\prod^{i=k}_{1}
			\hspace{-0.1cm}{\m{\Phi}^{i}_{i-1}}^{\hspace{-0.1cm}\top}}\hspace{-0.1cm} {\m{\Phi}^{i}_{i-1}},
\vspace{0.1cm}
	\end{align*}
where $\m{C}^{\top}\m{C} = \eye_{n_x} \in \mathbb{R}^{n_x \times n_x}$, considering a linear measurement model with sensors measuring all the state variables at each node. Such that, $\eye_{n_x}$ is parameterized with the sensor parameterization vector $\m{\gamma}$. The above holds true as a result of inner product multiplication $\m{\Psi}^{\top}\m{\Psi}$ being equivalent to $\sum_{k=0}^{\mr{N-1}}{\m{\Psi}^{k}_0}^{\top}\m{\Psi}^{k}_0$. Following this, and based on the max-type superposition property of LEs~\cite[Theorem~2.1.2]{Barreira1998}, we obtain the following by taking the $\log\det$ of $\m{V}_o$.
	\begin{align*} \log\det (\m{V}_{o})
		&\hspace{-0.08cm}=\hspace{-0.08cm}
\log\det \hspace{-0.1cm} \left(\sum_{k=0}^{\mr{N}-1}\; {\prod^{i=k}_{1}{\m{\Phi}^{i}_{i-1}}^{\top}}
		{\m{\Phi}^{i}_{i-1}}\right),\hspace{-0.1cm}\\
		&\hspace{-0.08cm}=\hspace{-0.08cm}
\log\det \hspace{-0.1cm} \left({\m{\Phi}_{0}^{\mr{N-1}}}^{\top}{\m{\Phi}_{0}^{\mr{N-1}}}\hspace{-0.05cm}\right)\hspace{-0.1cm},\\
		&\hspace{-0.08cm}=\hspace{-0.08cm}
\mr{log}\left(\prod_{i=1}^{n_x}\lambda_i\right) = \sum_{i=1}^{n_x} \mr{log}\lambda_i.\hspace{-0.2cm}
\end{align*}

	On a similar note, taking the $\mr{trace}\left( \lambda(\m{\Lambda}_{L})\right)$ for $k\rightarrow \mr{N} \approx \mr{N}-1$, we obtain 
	\begin{align*} \mr{trace}\left(\lambda \left(\frac{1}{k}\log \left(\left({\m{\Phi}_{0}^{k}}^{\top}\m{\Phi}_{0}^{k}\right)^{1/2}\right)\right)\right)
=
		\tfrac{1}{2\mr{N}} \mr{trace}(\m{\lambda}_{L}),
		\hspace{-0.2cm}
	\end{align*}
	such that, by evaluating $\mr{trace}( \m{\lambda}_{L})$, we obtain $\tfrac{1}{2\mr{N}}\sum_{i=1}^{n_x}\lambda_{L,i} = \sum_{i=1}^{n_x}\mr{log}\lambda_{i} $. This holds true given that $\m{\lambda}_{L}$ is a diagonal matrix, thereby completing the proof.
\end{proofofc}

\begin{proofofc}{prs:VarGramModular_prop}\label{proof:proofModular}
For any $\mathcal{S}\subseteq\mathcal{V}$, observe that
\begin{align*}
		\m{V}_{o}({\mathcal{S}})
		&\hspace{-0.05cm}=\hspace{-0.05cm}
		\bmat{\left\{\tfrac{\partial{\m{h}(\m{x}_{k})}}{\partial\m{x}_k}\m{\Phi}_{0}^{k}\right\}^{\mr{N-1}}_{k=0}}^{\top}
		\bmat{\left\{\tfrac{\partial{\m{h}(\m{x}_{k})}}{\partial\m{x}_k}\m{\Phi}_{0}^{k}\right\}^{\mr{N-1}}_{k=0}},\\
		&\hspace{-0.05cm}=\hspace{-0.05cm}
		\bmat{\bmat{\m I \otimes \m \Gamma \m C}\left\{\m{\Phi}_{0}^{k}\right\}^{\mr{N-1}}_{k=0}}^{\top}
		\bmat{\bmat{\m I \otimes \m \Gamma \m C}\left\{\m{\Phi}_{0}^{k}\right\}^{\mr{N-1}}_{k=0}},\\
&\hspace{-0.05cm}=\hspace{-0.05cm}
		\sum_{k=0}^{\mr{N}-1} \bmat{\prod^{i=k}_{1}
			\m{\Phi}^{i}_{i-1}}^{\top}\bmat{\m{C}^{\top} \Gamma^{2}\m{C}}
		\bmat{\prod^{i=k}_{1}\m{\Phi}^{i}_{i-1}},\hspace{-0.2cm}
\end{align*}
	where $\m \Gamma^2= \m \Gamma$, since it is a binary matrix. Now, denoting $\m c_j\in\mbb{R}^{1\times n_x}$ as the $j$-th row of $\m C$, then
	\begin{align*}
		\m{V}_{o}({\mathcal{S}})
&\hspace{-0.05cm}=\hspace{-0.05cm}
		\sum_{k=0}^{\mr{N}-1} \bmat{\prod^{i=k }_{1}
			\m{\Phi}^{i}_{i-1}}^{\top} \left(\sum_{j=1}^{n_y} \gamma_j\m c_j^\top \m c_j \right) \bmat{\prod^{i=k}_{1}
			\m{\Phi}^{i}_{i-1}},\hspace{-0.2cm}\\
		&\hspace{-0.05cm}=\hspace{-0.05cm}
		\sum_{j=1}^{n_y} \gamma_j \left(\sum_{k=0}^{\mr{N}-1} \bmat{\prod^{\substack{i=k}}_{1} \m{\Phi}^{i}_{i-1}}^{\top}  \bmat{\m c_j^\top \m c_j}
		\bmat{\prod^{\substack{i=k}}_{1} \m{\Phi}^{i}_{i-1}}\right),\\
		&\hspace{-0.05cm}=\hspace{-0.05cm}
		\sum_{j\in \mathcal{S}} \left\{{\m{\Phi}_{0}^{k}}^\top \hspace{-0.05cm}\m c_j^\top \m c_j \m{\Phi}_{0}^{k}\right\}^{\mr{N}-1}_{k=0} =  \sum_{j\in \mathcal{S}} \m{V}_{o}(j).\end{align*}
	
	The notation $j\in \mathcal{S}$ corresponds to every activated sensor such that $\gamma_j = 1$. This shows that $\m{V}_{o}(\mathcal{S})$ is a linear matrix function of $\gamma_j$ satisfying modularity as
	$\m{V}_{o}(\mathcal{S})= \m{V}_{o}(\emptyset) + \sum_{j\in\mathcal{S}} \m{V}_{o}(j)$, (Definition~\ref{def:modularity}). This concludes the proof.
\end{proofofc}

\vspace{-0.2cm}
\section{Nonlinear Output Mapping Functions}\vspace{-0.1cm}\label{apndx:generaloutput}
The following establishes how a generalized measurement function affects the derivation of certain equations within the proofs in this manuscript. The measurement equation can be written as 
\begin{align*}\label{prop:eq1}
	\m{\Psi} \hspace{-0.05cm}&=\hspace{-0.05cm}
	\left\{\tfrac{\partial{\m{h}(\m{x}_{k})}}{\partial\m{x}_k}\m{\Phi}_{0}^{k}\right\}^{\mr{N}-1}_{k=0}
	\hspace{-0.05cm}=\hspace{-0.05cm}
	\left\{\tfrac{\partial{\m{h}(\m{x}_{k})}}{\partial\m{x}_k}
	\prod^{i=k}_{1}\m{\Phi}^{i}_{i-1}\right\}^{\mr{N}-1}_{k=0}, \hspace{-0.2cm}\\
	\hspace{-0.05cm}&=\hspace{-0.05cm}
	\bmat{
	\big(\tfrac{\partial{\m{h}(\m{x}_{0})}}{\partial\m{x}_0}\m\Phi_0^0\big)^\top,
	\big(\tfrac{\partial{\m{h}(\m{x}_{1})}}{\partial\m{x}_1}\m\Phi_0^1\big)^\top,
		\ldots,
	\big(\tfrac{\partial{\m{h}(\m{x}_{\mr N-1})}}{\partial\m{x}_{\mr N-1}}\m\Phi_0^{\mr N-1}\big)^\top
	}^{\top}\hspace{-0.2cm},
\end{align*}
while noting that the multiplication of any matrix-valued vector with its transpose is equivalent to its matrix dot-product, the Var-Gram $\m V_o(\m x_0)=\m\Psi(\m x_0)^\top\m\Psi(\m x_0)$, over horizon $\mr{N}$, can be written as
\begin{equation}\label{eq:proof_prop_obs_matrix}
	\m V_o(\m x_0)
	= \sum_{k=0}^{\mr N-1} {\m\Phi_0^k}^\top
	\underbrace{\m C(\m x_k)^\top \m C(\m x_k)}_{=:~\tilde{\m C}(\m x_k)\succeq \m 0} \m\Phi_0^k,
\end{equation}
where $\m{\Phi}_{0}^{k} ={\big(\m{I}_{n_x}  +\tfrac{\partial\tilde{\m{f}}(\m{x}_{k})}{\partial\m{x}_{k}}\big)\tfrac{\partial\m{x}_{k}}{\partial\m{x}_{0}}}$ and ${\m{C}}(\m{x}_k) = \tfrac{\partial\m{h}(\m{x}_{k})}{\partial \m{x}_{k}}$. 
Thus, when constructing the Var-Gram, we obtain a multiplication of the differential of measurement mapping function~\eqref{eq:DT_measurement_model} with respect to $\m{x}_k$. The multiplication results in positive semidefinite symmetric matrices $\tilde{\m{C}}(\m{x}_k) = {\m{C}(\m{x}_k)}^{\top}\m{C}(\m{x}_k)$, where $\m{C}(\m{x}_k)$ is a constant measurement matrix that measures at most all the state variables (under full sensor selection). Thus,~\eqref{eq:proof_prop_obs_matrix} shows that considering a nonlinear output mapping function results in a trajectory dependent positive semidefinite matrix $\tilde{\m C}(\m x_k)\succeq 0$ that weights the variational mapping function $\m{\Phi}_0^k$ without scaling the exponential growth rates characterized by LEs~\cite[Theorem 2.1.2]{Barreira1998}, under the assumption of regularity.

\vspace{-0.1cm}
\balance
\bibliographystyle{IEEEtran}
\bibliography{references}
\end{document}

%% file: preamble.tex
\useRomanappendicesfalse

\usepackage{graphicx}
\usepackage{textcomp}
\usepackage{epsfig}
\usepackage{epsfig,color,amsmath,cite}

\usepackage{bm}
\usepackage{epstopdf}
\usepackage{amssymb}
\usepackage{url}

\usepackage{multirow}
\usepackage{hhline}
\usepackage{booktabs}
\usepackage[linesnumbered,boxed,commentsnumbered,ruled,vlined,longend]{algorithm2e}
\usepackage{comment}
\newcommand{\eps}{\varepsilon}

\DeclareMathOperator{\diag}{diag}

\DeclareMathOperator*{\maximize}{maximize}

\DeclareMathOperator*{\subjectto}{subject\ to}

\makeatother
\DeclareMathAlphabet\mathbfcal{OMS}{cmsy}{b}{n}



\usepackage{stackengine}

\newcommand{\mat}[1]{\boldsymbol{#1}}

\newcommand{\bmat}[1]{\begin{bmatrix} #1 \end{bmatrix}}

\providecommand{\eye}{\mat{I}}
\providecommand{\mA}{\ensuremath{\mat{A}}}

\newcommand{\m}{\boldsymbol}
\allowdisplaybreaks[4]
\usepackage[colorlinks = true,
linkcolor = blue,
urlcolor  = black,
citecolor = blue,
anchorcolor = blue]{hyperref}

\newcommand{\mb}[1]{\mathbf{#1}}
\newcommand{\mc}[1]{\mathcal{#1}}
\newcommand{\mbb}[1]{\mathbb{#1}}
\newcommand{\mr}[1]{\mathrm{#1}}
\usepackage[framemethod=TikZ]{mdframed}
\mdfdefinestyle{MyFrame}{linecolor=black,
	outerlinewidth=1.25pt,
	roundcorner=1.25pt,
	innerrightmargin=5pt,
	innerleftmargin=5pt,}

\usepackage[noabbrev]{cleveref}

\usepackage{mathtools}

\DeclarePairedDelimiter\abs{\lvert}{\rvert}\DeclarePairedDelimiter\norm{\lVert}{\rVert}
\makeatletter
\let\oldabs\abs
\def\abs{\@ifstar{\oldabs}{\oldabs*}}
\let\oldnorm\norm
\def\norm{\@ifstar{\oldnorm}{\oldnorm*}}
\makeatother

\usepackage[english]{babel}
\usepackage[utf8]{inputenc}
\usepackage[super]{nth}

\usepackage{graphicx}
\usepackage{float}
\usepackage[caption = false]{subfig}

\usepackage{array}
\usepackage{threeparttable}

\usepackage[english]{babel}
\usepackage[utf8]{inputenc}
\usepackage[super]{nth}

\usepackage{pifont}
\usepackage{mathtools}

\usepackage{mleftright}
\usepackage{xparse}

\NewDocumentCommand{\evaluat}{sO{\big}mm}{\IfBooleanTF{#1}
	{\mleft. #3 \mright|_{#4}}
	{#3#2|_{#4}}}

%% file: preamble_mkazma.tex
\usepackage{lettrine}
\usepackage{balance}

\usepackage{tikz}
\usetikzlibrary{shapes.geometric, arrows.meta, decorations.pathmorphing, decorations.markings, calc, 3d}
\usetikzlibrary{positioning, fit, backgrounds}
\usetikzlibrary{arrows.meta}  \tikzset{
	tangent/.style={decoration={
			markings,mark=at position #1 with {
				\coordinate (ta) at (0,0);
				\coordinate (tb) at (0.1,0);
			}
		},postaction=decorate},
	tangent/.default=0.5
}

\definecolor{Vgray}{RGB}{247,247,247}
\definecolor{Vedge}{RGB}{200,200,200}

\definecolor{Sub1purple}{RGB}{200,200,255}
\definecolor{Sub1edge}{RGB}{205,215,238}

\definecolor{Sub2yellow}{RGB}{251,234,206}
\definecolor{Sub2edge}{RGB}{251,234,206}
\definecolor{Sub3green}{RGB}{215,238,210}
\definecolor{Sub3edge}{RGB}{215,232,210}

\definecolor{Sub1color}{rgb}{0.85, 0.9, 1}
\definecolor{Sub2color}{rgb}{1, 0.9, 0.7}
\definecolor{Sub3color}{rgb}{0.8, 1, 0.8}
\definecolor{Sub4color}{RGB}{200,200,255}

\definecolor{Sub5edge}{RGB}{140,120,80}
\definecolor{Sub5color}{RGB}{255,245,225}

\definecolor{Sub6edge}{RGB}{100,130,150}
\definecolor{Sub6color}{RGB}{220,240,255}

\definecolor{Sub7edge}{RGB}{130,100,150}
\definecolor{Sub7color}{RGB}{245,230,255}

\definecolor{Sub8edge}{RGB}{150,130,100}
\definecolor{Sub8color}{RGB}{255,250,230}

\definecolor{Sub9edge}{RGB}{90,140,140}
\definecolor{Sub9color}{RGB}{220,250,250}

\definecolor{Sub10edge}{RGB}{110,110,110}
\definecolor{Sub10color}{RGB}{240,240,240}

\usepackage{pgfplots}
\pgfplotsset{compat=1.18}  \usepgfplotslibrary{fillbetween}

\definecolor{LightBlue}{rgb}{0.7, 0.85, 1}
\definecolor{LightOrange}{rgb}{1, 0.8, 0.6}
\definecolor{LightGreen}{rgb}{0.843, 0.933, 0.824}

\tikzstyle{box} = [rectangle, draw, thick, rounded corners,
text centered, font=\footnotesize, minimum height=1cm, text width=4.5cm]
\tikzstyle{colorbox1} = [rectangle, draw=black, rounded corners, fill=Sub1purple,
text centered, font=\footnotesize, minimum height=1cm, text width=4.5cm]
\tikzstyle{colorbox2} = [rectangle, draw=black, rounded corners, fill=Sub2yellow,
text centered, font=\footnotesize, minimum height=1cm, text width=4.5cm]
\tikzstyle{colorbox3} = [rectangle, draw=black, rounded corners, fill=Sub3green,
text centered, font=\footnotesize, minimum height=1cm, text width=4.5cm]
\tikzstyle{colorbox4} = [rectangle, draw=black, rounded corners, fill={rgb,255:red,247; green,247; blue,247},text centered, font=\scriptsize, minimum height=0.5cm, text width=3cm]

\tikzstyle{arrow} = [thick,->,>=stealth]
\tikzstyle{dashedarrow} = [thick, dashed, ->, >=stealth]

\let\labelindent\relax
\usepackage{enumitem}
\setlist[itemize]{leftmargin=*}

\definecolor{LCSS}{HTML}{004498}

\newcommand{\parstartc}[1]{\noindent \textbf{\textcolor{LCSS}{#1.}}\;}

\newcommand{\titlesc}[1]{\title{\Large \vspace{0.7cm} \LARGE \centering {\textsc{{#1}}}}}

\renewenvironment{proof}{\par\vspace{0.3em}\noindent\textbf{\textit{\textcolor{LCSS}{Proof.}}}\enspace}{\hfill $\blacksquare$\vspace{0.3em}}

\newenvironment{proofofc}[1]{\vspace{0.3em}\noindent\textbf{\textit{\textcolor{LCSS}{Proof of Theorem}}}~\textcolor{LCSS}{\textbf{\textit{\ref{#1}}.}}\enspace}{\hfill$\blacksquare$\vspace{0.3em}}

\newtheorem{theorem}{\textbf{\textcolor{LCSS}{Theorem}}}[section]

\newtheorem{mydef}{\textbf{\textcolor{LCSS}{Definition}}}[section]

\newtheorem{myrem}{\textbf{\textcolor{LCSS}{Remark}}}[section]
\newtheorem{asmp}{\textbf{\textcolor{LCSS}{Assumption}}}[section]
\newtheorem{mycor}{\textbf{\textcolor{LCSS}{Corollary}}}[section]
\newtheorem{myprs}{\textbf{\textcolor{LCSS}{Proposition}}}[section]

\makeatletter
\def\@thm#1#2{\refstepcounter{#1}\@ifnextchar[{\@ythm{#1}{#2}}{\@xthm{#1}{#2}}}
\def\@xthm#1#2{\@begintheorem{#2}{\csname the#1\endcsname}\ignorespaces}
\def\@ythm#1#2[#3]{\@opargbegintheorem{#2}{\csname the#1\endcsname}{#3}\ignorespaces}
\def\@endtheorem{\par}
\def\@begintheorem#1#2{\par\noindent\textbf{\textit{\textcolor{LCSS}{#1 #2.}}}\@ifnextchar[{\@withinfo}{\enspace\ignorespaces}}
\def\@withinfo[#1]{\textbf{\textcolor{LCSS}{ (#1)}}\enspace\ignorespaces}
\def\@opargbegintheorem#1#2#3{\par\noindent\textbf{\textit{\textcolor{LCSS}{#1 #2. (#3).}}}\enspace\ignorespaces}
\makeatother